\documentclass[prodmode]{small}
\usepackage[linesnumbered]{algorithm2e}

\SetAlFnt{\small}
\SetAlCapFnt{\small}
\SetAlCapNameFnt{\small}
\SetAlCapHSkip{0pt}
\IncMargin{-\parindent}

\acmYear{2020}
\acmMonth{01}
\usepackage{graphicx,amssymb}
\usepackage{color}
\usepackage{dsfont}

\usepackage{lastpage}
\usepackage{amsfonts}
\usepackage{amsmath}
\usepackage{tabularx}
\usepackage{multirow}
\usepackage[table]{xcolor}
\usepackage{booktabs}
\newcommand{\myvspace}[1]{}
\newcommand{\edge}[2]{\{#1,#2\}}
\newcommand{\fpt}{FPT}

\newcommand{\np}{NP}
\newcommand{\nph}{{\np}-hard}
\newcommand{\nphshort}{{\nph}}
\newcommand{\nphns}{{\np}-hardness}

\newcommand{\pnph}{para-{\nph}}

\newcommand{\wa}{W[1]}
\newcommand{\wah}{{\wa}-hard}
\newcommand{\wahns}{{\wa}-hardness}

\newcommand{\hide}[1]{}

\newcommand{\poly}{{P}}

\newcommand{\prob}[1]{{\sc{#1}}}

\newcommand{\discandi}{p} % distinguished candidate
\newcommand{\myfig}[1]{Figure~\ref{#1}}

\newcommand{\vset}[1]{V(#1)}
\newcommand{\eset}[1]{E(#1)}

\newcommand{\onlyfull}[1]{}

\newcommand{\bigo}[1]{O(#1)}
\newcommand{\bigos}[1]{O^*(#1)}
\newcommand{\yes}{Yes}
\newcommand{\no}{NO}
\newcommand{\yesins}{\yes-instance}
\newcommand{\noins}{\no-instance}
\newcommand{\sep}{{,}\ }
\newcommand{\abs}[1]{|#1|}
\newcommand{\setmid}{:}
\newcommand{\myit}[1]{#1}
\newcommand{\thm}{Theorem}

\newcommand{\xs}{U}
\newcommand{\xse}{c}
\newcommand{\xc}{S}
\newcommand{\xce}{s}
\newcommand{\xsize}{\kappa}
\newcommand{\EP}[3]{
\begin{center}
{\small
\begin{tabularx}{0.98\columnwidth}{ll}
\toprule
\multicolumn{2}{l}{\parbox[t]{0.9\columnwidth}{\textsc{#1}}} \\
\midrule
{\bf Input:}   & \parbox[t]{0.85\columnwidth}{#2\vspace*{1mm}}  \\
{\bf Question:}& \parbox[t]{0.85\columnwidth}{#3\vspace*{.5mm}} \\
\bottomrule
\end{tabularx}
}
\end{center}
}
\let\shortcite\cite

\begin{document}
\title{On the Complexity of Constructive Control under Nearly Single-Peaked Preferences}

\author
{
Yongjie Yang
\affil{Chair of Economic Theory, Faculty of Human and Business Sciences,\\ Saarland University, Saarbr\"{u}cken,
Germany\\ Email: yyongjiecs@gmail.com}
}

\begin{abstract}
We investigate the complexity of {\sc{Constructive Control by Adding/Deleting Votes}} (CCAV/CCDV) for $r$-approval, Condorcet, Maximin and Copeland$^{\alpha}$ in $k$-axes and $k$-candidates partition single-peaked elections. In general, we prove that CCAV and CCDV for most of the voting correspondences mentioned above are {\nph} even when~$k$ is a very small constant.
Exceptions are CCAV and CCDV for Condorcet and CCAV for $r$-approval in $k$-axes single-peaked elections, which we show to be fixed-parameter tractable with respect to~$k$.
%Due to relations among numerous concepts of nearly single-peaked domains, our {\nphns} results apply directly to many other domains.
%Whether CCDV for $r$-approval in $k$-axes single-peaked elections is fixed-parameter tractable with respect to~$k$ remains open.
In addition, we give a polynomial-time algorithm for recognizing $2$-axes elections, resolving an open problem. Our work leads to a number of dichotomy results.
To establish an {\nphns} result, we also study a property of $3$-regular bipartite graphs which may be of independent interest. In particular, we prove that for every $3$-regular bipartite graph,
there are two linear orders of its vertices such that the two endpoints of every edge are consecutive in at least one of the two orders.
\end{abstract}

\keywords{election control\sep complexity\sep single-peaked\sep NP-hard\sep Borda\sep Copeland\sep Maximin}

%\biblink{http://dblp.uni-trier.de/rec/bibtex/journals/tcs/YangG16}
\acmformat{Yongjie Yang. On the Complexity of Constructive Control under Nearly Single-Peaked Preferences. To appear in Proceedings of the 24th European Conference on Artificial Intelligence (ECAI 2020). 
%\href{...}{Click to obtain the bib}
}

\begin{bottomstuff}
This is the full version of the paper appeared in Proceedings of ECAI 2020 (\url{http://ecai2020.eu/}).

%Author's addresses: Y. Yang, Cluster of Excellence MMCI, Universit{\"a}t des Saarlandes.
\end{bottomstuff}
\maketitle
%\begingroup
%\let\clearpage\relax

%%%%%%%%%%%%%%%%%%end preamble full

\begingroup
\let\clearpage\relax
\section{Introduction}
\label{sec:introduction}

{\prob{Constructive Control by Adding Votes}} (CCAV) and {\prob{Constructive Control by Deleting Votes}} (CCDV) are two of the election control problems studied in the pioneering paper
by Bartholdi III, Tovey, and Trick~\shortcite{Bartholdi92howhard}.
These two problems model the applications where an election controller aims to make a distinguished candidate a winner by adding or deleting a limited number of voters.
Since their seminal work, the complexity of these problems under a lot of prestigious voting correspondences have been studied, and it turned out that many problems are
\nph~\cite{Bartholdi92howhard,DBLP:journals/jair/FaliszewskiHH11,DBLP:journals/jair/FaliszewskiHHR09,andrewlinphd2012}.
However, when restricted to single-peaked elections, many of them become polynomial-time solvable~\cite{BrandtBHH2015JAIRbypassingsinglepeakelectionAAAI10,DBLP:journals/iandc/FaliszewskiHHR11}.
Recall that an election is single-peaked if there is an order of the candidates, the so-called {\it{axis}}, such that each voter's preference purely increases, or decreases,
or first increases and then decreases along this order.
A natural question is that, as preferences of voters are extended from the single-peaked domain to the general domain with respect to a certain concept of nearly single-peakedness,
where does the complexity of these problems change? A large body of results have been reported with respect to some nearly single-peaked domains.
This paper aims to extend these study by investigating the complexity of the above two problems under several important voting correspondences when restricted to the
$k$-axes and $k$-candidates partition single-peaked elections ($k$-axes and $k$-CP elections for short respectively).
Generally speaking, an election is a $k$-axes election if there are~$k$ axes such that every vote is single-peaked with respect to at least one of the axes.
Equivalently, an election is a $k$-axes election if the votes can be partitioned into~$k$ sets each of which induces a single-peaked election.
An election is a {{$k$-CP election}} if there is a $k$-partition $(C_1,\dots,C_k)$ of the candidates such that the subelection restricted to each~$C_i$ is single-peaked.
Clearly, $1$-axis elections and $1$-CP elections are exactly single-peaked elections.
 %Thanks to the relations among many nearly-single peaked domains discovered by Erd\'{e}lyi,
 %Lackner, and Pfandler~\cite{Erdelyi2017}, many of our results apply to other domains of nearly single-peaked preferences. A detailed discussions on this is offered later.
The voting correspondences studied in this paper include $r$-approval, Condorcet, Maximin, and Copeland$^{\alpha}$, where $\alpha$ is a rational number such that~$0\leq \alpha\leq 1$.

Additionally, we also resolve an open question regarding the complexity of recognizing $2$-axes elections by deriving a polynomial-time algorithm.

\onlyfull{
\begin{figure*}
\centering{
\includegraphics[width=0.85\textwidth]{relation.pdf}
}
\caption{Parameters to measure closeness to the single-peaked domain and their relations studied in the literature.
The relations are due to~\protect\cite{Erdelyi2017}.
An arc from a parameter~$X$ to another parameter~$Y$ indicates that for any collection~$E$ of linear preferences, it holds that $X(E)\leq Y(E)$.}
\label{fig-relations}
\end{figure*}
}

\subsection{Related Work}
\onlyfull{In this section, we compare our main results with previous ones in the literature.}

Our study is clearly related to~\cite{BrandtBHH2015JAIRbypassingsinglepeakelectionAAAI10,DBLP:journals/iandc/FaliszewskiHHR11} where many voting problems including particularly
CCAV and CCDV for $r$-approval, Condorcet, Maximin, and Copeland$^1$ were shown to be polynomial-time solvable when restricted to single-peaked elections.
Resolving an open question, Yang~\shortcite{AAMAS17YangBordaSinlgePeaked} recently proved that CCAV and CCDV for Borda are {\nph} even when restricted to single-peaked elections.
\onlyfull{Particularly, Brandt~et~al.~\shortcite{BrandtBHH2015JAIRbypassingsinglepeakelectionAAAI10} proved that CCAV and CCDV for all weakCondorcet-consistent voting correspondences
are polynomial-time solvable when restricted to single-peaked elections. A voting correspondence is weakCondorcet-consistent if it selects exactly the weak Condorcet winners when
they exist~\footnote{We note that the notion of weakCondorcet-consistent has several versions studied in the literature(see~\cite{Felsenthal2014} for further discussions).
Our version here is the one used in~\cite{BrandtBHH2015JAIRbypassingsinglepeakelectionAAAI10}.
\onlyfull{ We shall discuss one of its variants later.}}.
It is known that Copeland$^1$ is weakCondorcet consistent. However, Copeland$^{\alpha}$ where $0\leq \alpha<1$ is not weakCondorcet consistent even in single-peaked elections.
Somewhat superiorly, to the best of our knowledge the complexity of CCAV and CCDV for Copeland$^{\alpha}$ where $0\leq \alpha<1$ in single-peaked elections has remained open to date.
We resolve this open question by giving a polynomial-time algorithm for the problem. To this end, we study several properties of Copeland$^{\alpha}$ winners in single-peaked elections
which may be of independent interest.
}

Our study is also closely related to the work of Yang and Guo~\shortcite{Yangaamas14a,DBLP:journals/jcss/YangG17,DBLP:journals/mst/YangG18}
where CCAV and CCDV for numerous voting correspondences in elections of single-peaked width at most~$k$ and $k$-peaked elections were studied.
Generally, an election has {\myit{single-peaked width~$k$}} if the candidates can be divided into groups, each of size at most~$k$,
such that every vote ranks all candidates in each group consecutively and, moreover, considering each group as a single candidate results in a single-peaked election.
An election is {\myit{$k$-peaked}} if there is an axis~$\lhd$ such that for every vote~$\pi$ there is a $k$-partition of~$\lhd$ such that~$\pi$ restricted to each component of the partition
is single-peaked. Obviously, $k$-CP elections are a subclass of $k$-peaked elections.
In addition, it is known that any election of single-peaked width~$k$ is a $k'$-CP election for some $k'\leq k$~\cite{Erdelyi2017}.
However, there are no general relation between $k$-axes elections and $k$-CP elections, and between $k$-axes elections and elections with single-peaked width~$k$~\cite{Erdelyi2017}.

In addition to CCAV and CCDV, many other problems restricted to single-peaked or nearly single-peaked domains have been extensively and intensively studied in the literature in the last decade
(see, e.g.,~\cite{DBLP:journals/jair/BetzlerSU13,DBLP:conf/ecai/CornazGS12,DBLP:conf/ijcai/CornazGS13,DBLP:conf/ijcai/YuCE13} for {\prob{Winner Determination}}, \cite{DBLP:conf/aaai/Walsh07}
for {\prob{Possible/Necessary Winner Determination}}, \cite{AAMAS15Yangmanipulationspwidth} for {\prob{Manipulation}}, \cite{DBLP:conf/aaai/MenonL16} for {\sc{Bribery}},
and~\cite{DBLP:journals/ai/FaliszewskiHH14,DBLP:conf/atal/Yang17} for some other important strategic voting problems).
Approval-Based multiwinner voting problems restricted to analogs of single-peaked domains have also been investigated from the
complexity perspective very recently~\cite{DBLP:conf/ijcai/ElkindL15,DBLP:conf/atal/LiuG16,DBLP:conf/aaai/Peters18}.

Finally, we point out that a parallel line of research on the complexity of single-crossing and nearly single-crossing domains has advanced immensely too
(see, e.g.,~\cite{DBLP:journals/aamas/MagieraF17,DBLP:journals/tcs/SkowronYFE15}).

We also refer to the book chapters~\cite{structuredpreferencesElkindLP,Hemaspaandra2016} and references therein for important development on these studies.

\onlyfull{
Many other problems pertaining to voting in nearly single-peaked elections have also been studied in the literature,
see, e.g.,~\cite{DBLP:conf/ecai/CornazGS12,DBLP:conf/ijcai/CornazGS13,AAMAS15Yangmanipulationspwidth,DBLP:conf/ijcai/YuCE13,DBLP:journals/mst/YangG18} and references therein for further details.
Moreover, voting problems in other restricted elections such as single-crossing elections have also been investigated recently, see, e.g.,~\cite{DBLP:journals/aamas/MagieraF17}.}

\subsection{Our Contributions}
Our contributions are summarized as follows.

\begin{itemize}\itemsep=2pt
\item We study CCAV and CCDV in $k$-axes and $k$-CP elections under $r$-approval, Condorcet, Copeland$^{\alpha}$, and Maximin.
\item We show that many problems already become {\nph} even when~$k$ is a very small constant. However, there are several exceptions.
(See Table~\ref{tab_our_results} below for the concrete results.)
In addition, our results reveal that from the parameterized complexity point of view, CCAV and CCDV for some voting correspondences behave completely differently.
For instance, for $r$-approval,  CCAV in $k$-axes elections is fixed-parameter tractable ({\fpt}) with respect to~$k$,
but CCDV is already {\nph} even for $k=2$, meaning that CCDV restricted to $k$-axes elections is even {\pnph} with respect to~$k$.
Our results also reveal that when restricted to different domains, the same problem may behave differently.
For instance, for Condorcet, we show that both CCAV and CCDV in $k$-axes elections are {\fpt} with respect to~$k$, but they become {\pnph} with respect to~$k$ when restricted to $k$-CP elections.
%We prove that CCAV for $r$-approval and Condorcet, and CCDV for Condorcet are {\fpt} with respect to~$k$.
%However, CCAV and CCDV for Maximin and Copeland$^{\alpha}\onlyfull{, 0\leq \alpha\leq 1}$\onlyfull{,} turn out to be {\nph} for every $k\geq 2$ and $0\leq \alpha\leq 1$.
Finally, we would like to point out that our study also leads to numerous dichotomy results for CCAV and CCDV with respect to the values of~$k$.
%We refer to Table~\ref{tab_our_results} for summaries of the results.
\item We study the complexity of determining whether an election is a $k$-axes election.
It is known that for $k=1$, the problem is polynomial-time solvable~\cite{Bartholdi1986T,Doignon:1994:PTA:182528.182531,DBLP:conf/ecai/EscoffierLO08}.
Erd\'{e}lyi, Lackner, and Pfandler~\shortcite{Erdelyi2017} proved that the problem is {\nph} for every $k\geq 3$.
We complement these results by showing that determining whether an election is a $2$-axes election is polynomial-time solvable,
filling the last complexity gap of the problem with respect to~$k$.
\end{itemize}

\renewcommand\arraystretch{1.2}
\begin{table}
\begin{center}
%\scalebox{0.85}
{
\begin{tabular}{|l|c|c|c|c|c|}\hline
 \multicolumn{6}{|c|}{\prob{Constructive Control by Adding Votes} (CCAV)}\\ \hline
&  SP & $(k\geq 2)$-Axes  &  $2$-CP &$(k\geq 3)$-CP & general
% &  SP & $(k\geq 2)$-Axes  & $2$-CP & $(k\geq 3)$-CP & general%
 \\ \hline

\multirow{2}{*}{$r$-approval} &
\multirow{2}{*}{\poly~$\clubsuit$} &
\multirow{2}{*}{{\bf{\fpt}}{ ({\thm}~\ref{thm_r_approval_k_axis_fpt})}}&
\multirow{2}{*}{{\poly} $\diamondsuit$}&
$r\leq 3$: {\poly~$\heartsuit$} &
$r\leq 3$: {\poly~$\heartsuit$}
%&
%\multirow{2}{*}{\poly~\protect\cite{DBLP:journals/iandc/FaliszewskiHHR11}} &
%%$r\leq 2$: {\poly} &
%\multicolumn{3}{c|}{$r\leq 2$: {\poly}~\protect\cite{DBLP:conf/icaart/Lin11}}&
%$r\leq 2$: {\poly~\protect~\cite{DBLP:conf/icaart/Lin11}}%
\\

&
&
&
&
{{$r\geq 4$: \bf\nphshort}{ ({\thm}~\ref{thm_ccav_r_approval_3_cp_nph})}} &
$r\geq 4$: {\nphshort~$\heartsuit$}
%& &\multicolumn{3}{c|}{$r\geq 3$: {\bf\nphshort}{ ({\thm}~\ref{thm_ccdv_r_approval_2_cp_np_hard})}}&$r\geq 3$: {\nphshort~\protect\cite{DBLP:conf/icaart/Lin11}}
\\ \hline

Borda   &
\multicolumn{4}{c|}{{\nphshort}~$\spadesuit$} &
{\nphshort} $\maltese$
%&\multicolumn{4}{c|}{{\nphshort}~\protect\cite{AAMAS17YangBordaSinlgePeaked}} &{\nphshort}~\cite{DBLP:journals/tcs/LiuZ13}%
 \\ \hline

{Condorcet} &
{\poly~$\sharp$} &
{\bf\fpt}{ ({\thm}~\ref{thm_Condorcet_FPT_k_dimension})} &
\onlyfull{{\bf{\poly}} } open&
{\bf\nphshort} ({\thm}~\ref{thm-CCAV-CCDV-Condorcet-3-CP-NP-hard})&
{\nphshort~$\natural$}
%&{\poly~\protect\cite{BrandtBHH2015JAIRbypassingsinglepeakelectionAAAI10}} &{\bf\fpt}{  ({\thm}~\ref{thm_Condorcet_FPT_k_dimension})}
%&\onlyfull{{\bf{\poly}}} open &{\bf\nphshort} ({\thm}~\ref{thm-CCAV-CCDV-Condorcet-3-CP-NP-hard})&{\nphshort~\protect\cite{Bartholdi92howhard}}%
\\ \hline

Copeland$^{\alpha\in [0,1)}$ &
open &
{\bf{\nphshort}} ({\thm}~\ref{thm_CCAV_CCDV_Copeland_k_additional_Axis_NP_hard})&
\multicolumn{2}{c|}{\nphshort~$\yen$} &
{\nphshort}~$\P$
%&open &{\bf{\nphshort}} ({\thm}~\ref{thm_CCAV_CCDV_Copeland_k_additional_Axis_NP_hard})&\multicolumn{2}{c|}{\nphshort~\protect\cite{Yangaamas14a}}
%&{\nphshort}~\protect\cite{DBLP:journals/jair/FaliszewskiHHR09}%
\\ \hline

Copeland$^1$&
{\poly}~$\sharp$ &
{\bf{\nphshort}} ({\thm}~\ref{thm_CCAV_CCDV_Copeland_k_additional_Axis_NP_hard})&
\multicolumn{2}{c|}{{\bf{\nphshort}} ({\thm}~\ref{thm-CCAV-CCDV-Copeland-1-maximin-2-cp-np-hard})}&
{\nphshort}~$\P$
%&{\poly~\protect\cite{BrandtBHH2015JAIRbypassingsinglepeakelectionAAAI10}} &{\bf{\nphshort}} ({\thm}~\ref{thm_CCAV_CCDV_Copeland_k_additional_Axis_NP_hard})&
%\multicolumn{2}{c|}{{\bf{\nphshort}} ({\thm}~\ref{thm-CCAV-CCDV-Copeland-1-maximin-2-cp-np-hard})}&
%{\nphshort}~\protect\cite{DBLP:journals/jair/FaliszewskiHHR09}%
\\ \hline

Maximin &
{\poly}~$\sharp$ &
{\bf{\nphshort}} ({\thm}~\ref{thm_CCAV_CCDV_Copeland_k_additional_Axis_NP_hard})&
\multicolumn{2}{c|}{{\bf{\nphshort}} ({\thm}~\ref{thm-CCAV-CCDV-Copeland-1-maximin-2-cp-np-hard})}&
{\nphshort}~$\S$
%&{\poly~\protect\cite{BrandtBHH2015JAIRbypassingsinglepeakelectionAAAI10}} &{\bf{\nphshort}} ({\thm}~\ref{thm_CCAV_CCDV_Copeland_k_additional_Axis_NP_hard})&
%\multicolumn{2}{c|}{{\bf{\nphshort}} ({\thm}~\ref{thm-CCAV-CCDV-Copeland-1-maximin-2-cp-np-hard})}&
%{\nphshort}~\protect\cite{DBLP:journals/jair/FaliszewskiHH11}
\\ \hline
\end{tabular}
}
\bigskip

{
\begin{tabular}{|l|c|c|c|c|c|}\hline
 \multicolumn{6}{|c|}{\prob{Constructive Control by Deleting Votes} (CCDV)}\\ \hline

%&  SP & $(k\geq 2)$-Axes  &  $2$-CP &$(k\geq 3)$-CP & general
&
SP & $(k\geq 2)$-Axes  & $2$-CP & $(k\geq 3)$-CP & general \\ \hline

\multirow{2}{*}{$r$-approval} &
%\multirow{2}{*}{\poly~\protect\cite{DBLP:journals/iandc/FaliszewskiHHR11}} &
%\multirow{2}{*}{{\bf{\fpt}}{ ({\thm}~\ref{thm_r_approval_k_axis_fpt})}}&
%\multirow{2}{*}{{\poly} \protect\cite{DBLP:journals/jcss/YangG17}}&
%$r\leq 3$: {\poly~\protect\cite{DBLP:conf/icaart/Lin11}} &
%$r\leq 3$: {\poly~\protect\cite{DBLP:conf/icaart/Lin11}} &
\multirow{2}{*}{\poly~$\clubsuit$} &
%$r\leq 2$: {\poly} &
\multicolumn{3}{c|}{$r\leq 2$: {\poly}~$\heartsuit$}&
$r\leq 2$: {\poly~$\heartsuit$}\\

%&
%&
%&
%&
%{{$r\geq 4$: \bf\nphshort}{ ({\thm}~\ref{thm_ccav_r_approval_3_cp_nph})}} &
%$r\geq 4$: {\nphshort~\protect\cite{DBLP:conf/icaart/Lin11}}
& &
\multicolumn{3}{c|}{$r\geq 3$: {\bf\nphshort}{ ({\thm}~\ref{thm_ccdv_r_approval_2_cp_np_hard})}}&
$r\geq 3$: {\nphshort~$\heartsuit$}\\ \hline

Borda   &
%\multicolumn{4}{c|}{{\nphshort}~\protect\cite{AAMAS17YangBordaSinlgePeaked}} &
%{\nphshort} \protect\cite{Nathan07} &
\multicolumn{4}{c|}{{\nphshort}~$\spadesuit$} &
{\nphshort}~$\pounds$ \\ \hline

{Condorcet} &
%{\poly~\protect\cite{BrandtBHH2015JAIRbypassingsinglepeakelectionAAAI10}} &
%{\bf\fpt}{ ({\thm}~\ref{thm_Condorcet_FPT_k_dimension})} &
%\onlyfull{{\bf{\poly}} } open&
%{\bf\nphshort} ({\thm}~\ref{thm-CCAV-CCDV-Condorcet-3-CP-NP-hard})&
%{\nphshort~\protect\cite{Bartholdi92howhard}} &
{\poly~$\sharp$} &
{\bf\fpt}{  ({\thm}~\ref{thm_Condorcet_FPT_k_dimension})} &
\onlyfull{{\bf{\poly}}} open &
{\bf\nphshort} ({\thm}~\ref{thm-CCAV-CCDV-Condorcet-3-CP-NP-hard})&
{\nphshort~$\natural$}\\ \hline

Copeland$^{\alpha\in [0,1)}$ &
%open &
%{\bf{\nphshort}} ({\thm}~\ref{thm_CCAV_CCDV_Copeland_k_additional_Axis_NP_hard})&
%\multicolumn{2}{c|}{\nphshort~\protect\cite{Yangaamas14a}} &
%{\nphshort}~\protect\cite{DBLP:journals/jair/FaliszewskiHHR09}&
open &
{\bf{\nphshort}} ({\thm}~\ref{thm_CCAV_CCDV_Copeland_k_additional_Axis_NP_hard})&
\multicolumn{2}{c|}{\nphshort~$\yen$}&
{\nphshort}~$\P$\\ \hline

Copeland$^1$&
%{\poly}~\protect\cite{BrandtBHH2015JAIRbypassingsinglepeakelectionAAAI10} &
%{\bf{\nphshort}} ({\thm}~\ref{thm_CCAV_CCDV_Copeland_k_additional_Axis_NP_hard})&
%\multicolumn{2}{c|}{{\bf{\nphshort}} ({\thm}~\ref{thm-CCAV-CCDV-Copeland-1-maximin-2-cp-np-hard})}&
%{\nphshort}~\protect\cite{DBLP:journals/jair/FaliszewskiHHR09}&
{\poly~$\sharp$} &
{\bf{\nphshort}} ({\thm}~\ref{thm_CCAV_CCDV_Copeland_k_additional_Axis_NP_hard})&
\multicolumn{2}{c|}{{\bf{\nphshort}} ({\thm}~\ref{thm-CCAV-CCDV-Copeland-1-maximin-2-cp-np-hard})}&
{\nphshort}~$\P$\\ \hline

Maximin &
%{\poly}~\protect\cite{BrandtBHH2015JAIRbypassingsinglepeakelectionAAAI10} &
%{\bf{\nphshort}} ({\thm}~\ref{thm_CCAV_CCDV_Copeland_k_additional_Axis_NP_hard})&
%\multicolumn{2}{c|}{{\bf{\nphshort}} ({\thm}~\ref{thm-CCAV-CCDV-Copeland-1-maximin-2-cp-np-hard})}&
%{\nphshort}~\protect\cite{DBLP:journals/jair/FaliszewskiHH11}&
{\poly~$\sharp$} &
{\bf{\nphshort}} ({\thm}~\ref{thm_CCAV_CCDV_Copeland_k_additional_Axis_NP_hard})&
\multicolumn{2}{c|}{{\bf{\nphshort}} ({\thm}~\ref{thm-CCAV-CCDV-Copeland-1-maximin-2-cp-np-hard})}&
{\nphshort}~$\S$\\ \hline
\end{tabular}
}
\caption{A summary of the complexity of CCAV and CCDV in general and nearly single-peaked domains.
Here, ``P'' stands for ``polynomial-time solvable''. %and ``{\nphshort}'' for ``{\nph}''.
Moreover, ``SP'' stands for ``single-peaked''.
Our results are in boldface. All {\fpt} results are with respect to~$k$.
Results marked by~$\clubsuit$ are from~\protect\cite{DBLP:journals/iandc/FaliszewskiHHR11},
by~$\diamondsuit$ from~\protect\cite{DBLP:journals/jcss/YangG17},
by~$\heartsuit$  from~\protect\cite{DBLP:conf/icaart/Lin11},
by~$\spadesuit$ from~\protect\cite{AAMAS17YangBordaSinlgePeaked},
by~$\maltese$ from~\protect\cite{Nathan07},
by~$\sharp$ from~\protect\cite{BrandtBHH2015JAIRbypassingsinglepeakelectionAAAI10},
by~$\P$ from \protect\cite{DBLP:journals/jair/FaliszewskiHHR09},
by~$\natural$ from \protect\cite{Bartholdi92howhard},
by~$\S$ from \protect\cite{DBLP:journals/jair/FaliszewskiHH11},
by~$\pounds$ from \protect\cite{DBLP:journals/tcs/LiuZ13},
and by~$\yen$ from \protect\cite{Yangaamas14a}.
%The two {\fpt} results for Condorcet hold only when~$k$ axes of the given election are given.
%The complexity of CCDV for $r$-approval in $k$-axis elections remains open for every $r\geq 3$ and every constant $k\geq 2$.
%The {\nphns} of CCAV and CCDV for Copeland$^{\alpha}$ where $0\leq \alpha<1$ in $k$-PC elections are implied by the results in~\protect\cite{Yangaamas14a}.
}
\label{tab_our_results}
\end{center}
\end{table}

%%%%%%%%%%%%%%%%%%%%%%%%%%%%%%%%%%%
%%%%%%%%%%%%%%%%%%%%%%%%%%%%%

%Due to space limitation, several proofs are omitted in this version.

\section{Preliminaries}
In this section, we give the notions used in the paper.
For a positive integer~$i$, we use~$[i]=\{j\in \mathbb{N} \setmid 0<j\leq i\}$ to denote the set of all positive integers no greater than~$i$.
\smallskip

{\bf{Election.}} An {\it{election}} is a tuple $\mathcal{E}=(\mathcal{C},\Pi_{\mathcal{V}})$,
where~$\mathcal{C}$ is a set of {\it{candidates}} and~$\Pi_{\mathcal{V}}$ a multiset of {\it{votes}},
defined as permutations (linear orders) over~$\mathcal{C}$. 
For two candidates $c, c'\in \mathcal{C}$ and a vote $\pi\in \Pi_{\mathcal{V}}$, we say~$c$ is {\it{ranked above}}~$c'$ or~$\pi$ {\it{prefers}}~$c$ to~$c'$ if $\pi(c)< \pi(c')$. 
For two subsets $X,Y\subseteq \mathcal{C}$ of candidates, a vote with preference $X\succ Y$ means that this vote prefers every $x\in X$ to every $y\in Y$.
For brevity, we use $x\succ y$ for $\{x\}\succ \{y\}$.  
Here,~$\pi(c)$ is the {\it{position}} of~$c$ in~$\pi$, i.e., $\pi(c)=\abs{\{c'\in \mathcal{C} \setmid \pi(c')<\pi(c)\}}+1$.
For $C\subseteq \mathcal{C}$ and a vote $\pi\in \Pi_{\mathcal{V}}$, let $\pi(C)=\{\pi(c) \mid c\in C\}$.
In addition, let~$\pi^C: C\rightarrow [\abs{C}]$ be~$\pi$ restricted to~$C$ so that for~$c,c'\in C$, $\pi(c)<\pi(c')$ implies $\pi^C(c)<\pi^C(c')$.
Let $\Pi^C_{\mathcal{V}}=\{\pi^C \mid \pi\in \Pi_{\mathcal{V}}\}$.
Hence, $(C, \Pi^C_{\mathcal{V}})$ is the election $(\mathcal{C}, \Pi_{\mathcal{V}})$ restricted to~$C$.
We use $N_{\mathcal{E}}(c,c')$ to denote the number of votes preferring~$c$ to~$c'$ in~$\mathcal{E}$.
We drop~$\mathcal{E}$ from the notation when it is clear from the context which election is considered.
For two candidates~$c$ and~$c'$ in~$\mathcal{C}$, we say~{\it{$c$ beats~$c'$}} if $N(c,c')>N(c',c)$, and~$c$ {\it{ties}}~$c'$ if $N(c,c')=N(c',c)$.

For a linear order~$\lhd$ over a set~$A$, we say that two elements in~$A$ are {\it{consecutive}} if there are no other elements from~$A$ between them in the order.

\smallskip

{\bf{Voting correspondence.}}
A {\it{voting correspondence}}~$\varphi$ is a function that maps an election $\mathcal{E}=(\mathcal{C},\Pi_{\mathcal{V}})$ to a
non-empty subset~$\varphi(\mathcal{E})$ of~$\mathcal{C}$.
We call the elements in~$\varphi(\mathcal{E})$ the {\it{winners}} of~$\mathcal{E}$ with respect to~$\varphi$.
\onlyfull{If~$\varphi(\mathcal{E})$ contains only one winner, we call the winner the {\it{unique winner}}; otherwise, we call them {\it{nonunique winners}}.}
 The following voting correspondences are related to our study.
%{\bf{Voting Correspondences.}} %Let $(\mathcal{C}, \Pi_{\mathcal{V}})$ be an election and $m=|\mathcal{C}|$.
%For other voting correspondences mentioned in this paper, we refer to~\cite{DBLP:conf/birthday/BetzlerBCN12,DBLP:conf/atal/HemaspaandraLM13}.
%\medskip

\begin{description}\itemsep=5pt%[labelindent=16pt]
\item[$r$-approval] Each vote approves exactly its top-$r$ candidates. Winners are those with the most total approvals.
Throughout this paper,~$r$ is assumed to be a constant, unless stated otherwise.

\item[Borda] Every vote~$\pi$ gives $m-\pi(c)$ points to every candidate~$c$ and the winners are the ones with the highest total score. Here,~$m$ is the number of candidates.

%\item[Condorcet] If there is a candidate~$c$ who beats every other candidate,~$c$ wins; otherwise, all candidates win.

\item[Copeland$^\alpha$ ($0\leq \alpha \leq 1$)]
%{\bf{Copeland$^\alpha$
For a candidate~$c$, let~$B(c)$ (resp.\ $T(c)$) be the set of candidates beaten by~$c$ (resp.\ tie with~$c$). The Copeland$^\alpha$ score of~$c$ is $|B(c)|+\alpha\cdot |T(c)|$.
A Copeland$^{\alpha}$ winner is a candidate with the highest score.

\item[Maximin]
%{\bf{Maximin:}}
The Maximin score of a candidate~$c\in \mathcal{C}$
is $\min_{c'\in \mathcal{C}\setminus \{c\}}N(c,c')$. Maximin winners are those with the highest Maximin score.

\item[Condorcet] The {\it{Condorcet winner}} of an election is the candidate that beats all other candidates.
It is well-known that each election has either zero or exactly one Condorcet winner. In addition, both Copeland$^{\alpha}$ and Maximin select only the Condorcet winner if it exists.
In this paper, the Condorcet correspondence refers to as the following one: if the Condorcet winner exists, it is the unique winner; otherwise, all candidates win.
\end{description}

Our discussion also needs the concept of weak Condorcet winner. Precisely, a candidate in an election is a weak Condorcet winner if and only if it is not beaten by anyone else.

\smallskip
{\bf{Nearly single-peaked elections.}} An election~$(\mathcal{C}, \Pi_{\mathcal{V}})$ is {\textit{single-peaked}}
if there is a linear order~$\lhd$ of~$\mathcal{C}$, called an {\it{axis}},
such that for every vote $\pi\in \Pi_{\mathcal{V}}$ and every three candidates $a,b,c\in\mathcal{C}$ with $a\;\lhd\;b\;\lhd\;c$ or $c\;\lhd\;b\;\lhd\;a$,
it holds that $\pi(c)< \pi(b)$ implies $\pi(b)<\pi(a)$. An election $(\mathcal{C}, \Pi_{\mathcal{V}})$ is {\it{$k$-axes single-peaked}}
if there are~$k$ axes $\lhd_1,\dots,\lhd_k$ such that every $\pi\in \Pi_{\mathcal{V}}$ is single-peaked with respect to at least one of $\{\lhd_1,\dots,\lhd_k\}$.
\onlyfull{For simplicity, we call such an election a {\it{$k$-axes election}}.}
In addition, $(\mathcal{C}, \Pi_{\mathcal{V}})$ is a {\it{$k$-CP election}}
if there is a $k$-partition $(C_1,\dots,C_k)$ of~$\mathcal{C}$ such that for all $i\in [k]$, $(C_i, \Pi^{C_i}_{\mathcal{V}})$ is single-peaked.
\onlyfull{For simplicity, we call such an election a {\it{$k$-CP election}}.}

%In what follows, general elections are elections where the domain of the votes is not restricted, i.e., each vote can be any permutation over the set of candidates.

\medskip
{\bf{Problem formulation.}}
For a voting correspondence~$\varphi$, we study the following two problems.

{\EP{Constructive Control by Adding Votes (CCAV)}
{An election $(\mathcal{C}, \Pi_{\mathcal{V}})$, a distinguished candidate $p\in \mathcal{C}$, a multiset~$\Pi_{\mathcal{W}}$ of votes, and a positive integer~$\ell$.}
{Is there $\Pi_{W}\subseteq \Pi_{\mathcal{W}}$ such that $|\Pi_{W}|\leq \ell$ and~${\discandi}$ uniquely wins
$(\mathcal{C}, \Pi_{\mathcal{V}}\cup \Pi_{W})$ with respect to~$\varphi$?}
}

{In the above definition, votes in~$\Pi_{\mathcal{V}}$ and~$\Pi_{\mathcal{W}}$ are referred to as {\it{registered votes}} and {\it{unregistered votes}}, respectively.}
For an instance $((\mathcal{C}, \Pi_{\mathcal{V}}), p\in \mathcal{C}, \Pi_{\mathcal{W}}, \ell)$ of CCAV,
a subset $\Pi_{W}\subseteq \Pi_{\mathcal{W}}$ is called a {\myit{feasible solution}}
of the instance
if~${\discandi}$ uniquely wins $(\mathcal{C}, \Pi_{\mathcal{V}}\cup \Pi_{W})$.

{\EP{Constructive Control by Deleting Votes (CCDV)}
{An election $(\mathcal{C}, \Pi_{\mathcal{V}})$, a distinguished candidate $p\in \mathcal{C}$, and a positive integer~$\ell$.}
{Is there $\Pi_{V}\subseteq \Pi_{\mathcal{V}}$ such that $|\Pi_{V}|\leq \ell$ and~${\discandi}$ uniquely wins the election $(\mathcal{C}, \Pi_{\mathcal{V}}\setminus  \Pi_{V})$
with respect to~$\varphi$?}
}

For an instance $((\mathcal{C}, \Pi_{\mathcal{V}}), p\in \mathcal{C}, \ell)$ of CCAV, a subset $\Pi_{V}\subseteq \Pi_{\mathcal{V}}$ is called a {\myit{feasible solution}}
of the instance
if~${\discandi}$ uniquely wins $(\mathcal{C}, \Pi_{\mathcal{V}}\setminus \Pi_{V})$.
An optimal solution of a {\yesins} of CCAV/CCDV refers to as a feasible solution consisting of the minimum votes.

In this paper, we study CCAV and CCDV in $k$-CP ($k$-axes) elections. For CCAV, we mean that $(\mathcal{C}, \Pi_{\mathcal{V}}\cup \Pi_{\mathcal{W}})$
is a $k$-CP ($k$-axes) election.
\onlyfull{Moreover, for these problems in $k$-axes elections, we assume that the~$k$ axes with respect to which each vote is single-peaked with respect
to at least one of them are given.
Similarly, for $k$-CP elections, we assume that the $k$-partition of~$\mathcal{C}$, restricted to each component of which the votes are single-peaked is given.
Nevertheless, all {\nphns} results hold even when these information are unknown.
The above assumptions make sense as in many real-world applications, the~$k$ axes (resp.\ $k$-partition of~$\mathcal{C}$) are known in advance.
This is actually one of the reasons why domain restricted elections can arise in practice.
For example, in real-world single-peaked political elections, the voters are thought to agree upon that the candidates are ordered on a common known left-right dimension.
}

\smallskip

%{\bf{An {\nph}} problem.} All our {\nphns} results are by reductions from the following problem.

For {\nphns} results, we are only interested in the minimum values of~$k$ for which CCAV and CCDV in $k$-CP ($k$-axes) elections are {\nph}.
In fact, one can easily show that, for all voting correspondences considered in this paper,
if CCAV and CCDV in $k$-CP (resp.\ $k$-axes) elections are {\nph}, so are they in $(k+1)$-CP (resp.\ $(k+1)$-axes) elections.
\onlyfull{First, there is a trivial reduction from CCAV/CCDV for all aforementioned voting correspondences in $k$-CP elections to CCAV/CCDV in $(k+1)$-CP elections:
given an instance with a $k$-CP election, create one more candidate ranked in the last position in all votes. Moreover, if a $k$-axes election contains at least $k+1$ votes
(this is the case in our proofs), it is also a $(k+1)$-axes election.}
 Our {\nphns} results are based on reductions from the following problem.

 \EP
{Restricted Exact Cover by $3$-Sets (RX3C)}
{A universe $\xs=\{\xse_1,\xse_2,\dots,\xse_{3{\xsize}}\}$ and a collection $\xc=\{\xce_1,\xce_2,\dots,\xce_{3\xsize}\}$ of $3$-subsets of~$\xs$ such that each~$\xse_i\in \xs$
occurs in exactly three elements of~$\xc$.}
{Is there an $\xc'\subseteq \xc$ such that $|\xc'|={\xsize}$ and each $\xse_i\in \xs$ appears in exactly one element of~$\xc'$?}

%We assume that each element~$c_i\in \xs$ occurs in exactly~$3$ subsets of~${S}$.
%This assumption does not affect the {\nphns}\onlyfull{ of the problem}
The RX3C problem is {\nph}~\cite{DBLP:journals/tcs/Gonzalez85}.
%Observe that under this assumption, it holds that $n=3\xsize$.
%In all proofs of theorems stating {\nph} results, $I=(U=\{c_1,\dots,c_{3\xsize}\}, S=\{s_1,\dots,s_{3\xsize}\})$ is an RX3C instance.
\smallskip

 {\bf{Parameterized complexity.}} A {\it{parameterized problem instance}} is a tuple $(I,\kappa)$, where~$I$ denotes the main part and~$\kappa$
 is a parameter which is often an integer.
 A parameterized problem is {\fpt} if any of its instance $(I, \kappa)$ can be determined in $\bigos{f(\kappa)\cdot |I|^{\bigo{1}}}$ time,
 where~$f$ is a computable function and~$|I|$ is the size of the main part.
% A parameterized problem is in {\xp} if any of its instance $(I, \kappa)$ can be determined in $\bigos{|I|^{g(\kappa)}}$ time where~$g$ is a computable function.
  A parameterized problem is {\pnph} if there is a constant~$c$ such that the problem is {\nph} for any parameter greater than~$c$.
 For more detailed introduction to parameterized complexity, we refer to~\cite{DBLP:conf/coco/DowneyF92,DBLP:conf/dagstuhl/DowneyF92}.

\section{$\lowercase{r}$-Approval}
%Unless stated otherwise,~$r$ is a constant.
In general elections, CCAV and CCDV for $r$-approval are {\nph} even when~$r$ is a constant ($r\geq 4$ for CCAV and $r\geq 3$ for CCDV)~\cite{andrewlinphd2012}.
However, when restricted to single-peaked elections, both problems become polynomial-time solvable (even when~$r$ is not a constant)~\cite{DBLP:journals/iandc/FaliszewskiHHR11}.
%Yang and Guo~\shortcite{Yangaamas14b} showed that the {\nphns} remains in 2-peaked elections. However, CCAV for $r$-approval is {\nph} only when~$r$ is not a constant.
%For any constant~$r$, they developed a dynamic-based polynomial-time algorithm.
%In the same paper, Yang and Guo investigated the parameterized complexity the problems in general elections and showed that both problems are {\fpt} with respect to the solution size~$R$,
%for constant~$r$.
%
We complement these results by first showing that CCAV for $r$-approval in $k$-axes elections is {\fpt} with respect to~$k$.
%This also implies that CCAV and CCDV for $r$-approval is {\fpt} with respect to one of the single parameter in $\{k, r\}$ when the other one is a constant.
Our {\fpt}-algorithm is based on the following two observations

\begin{observation}
\label{obs-r-approval-candidates-consecutive-single-peaked}
If a vote is single-peaked with respect to an axis~$\lhd$, then all approved candidates in the vote lie consecutively in~$\lhd$.
\end{observation}

\begin{observation}
\label{obs-ccav-yes-instance-solution-approving-distinguished-candidate}
For every {\yesins} of the CCAV problem, any optimal solution consists of only unregistered votes approving the distinguished candidate.
\end{observation}

The above observations suggest that to solve an instance, we need only to focus on a limited number of candidates---the candidates at most~``$r$ far away''
from the distinguished candidate~$p$ in the $k$-axes of the given instance.
%Then, from the fact that CCAV for $r$-approval is {\fpt} with respect to the number of candidates~\cite{DBLP:conf/ecai/Yang14}, our result follows.
%Notice that our algorithm does not need to know that real~$k$ axes of the election.

\begin{theorem}
\label{thm_r_approval_k_axis_fpt}
CCAV for $r$-approval in $k$-axes elections is {\fpt} with respect to the combined parameter~$k+r$.
\end{theorem}

\begin{proof}
Let~$\mathcal{C}$,~$\Pi_{\mathcal{V}}$,~$\Pi_{\mathcal{W}}$,~$\discandi \in \mathcal{C}$,~$\ell$ be the components of the input of a CCAV instance as in the definition\onlyfull{ of the problem},
where $(\mathcal{C}, \Pi_{\mathcal{V}}\cup \Pi_{\mathcal{W}})$ is a $k$-axes election.
For each $c\in \mathcal{C}$, let~${\sf{sc}}(c)$ be the score of~$c$ with respect to~$\Pi_{\mathcal{V}}$, i.e.,~${\sf{sc}}(c)$ is the number of votes in~$\Pi_{\mathcal{V}}$ approving~$c$.
Let~$\Pi_{\discandi}$ be the multiset of all votes $\pi\in \Pi_{\mathcal{W}}$ such that $\pi(\discandi)\leq r$.
For each vote $\pi\in \Pi_{\discandi}$, let~$C(\pi)$ be the set of candidates ranked in the top-$r$ positions, i.e., $C(\pi)=\{c\in \mathcal{C} \mid \pi(c)\leq r\}$.
Moreover, let $B=\bigcup_{\pi\in \Pi_{\discandi}}C(\pi)\setminus \{\discandi\}$.
Due to Observation~\ref{obs-ccav-yes-instance-solution-approving-distinguished-candidate}, any optimal solution consists of only votes from~$\Pi_{\discandi}$.
Moreover, adding a vote in~$\Pi_{\discandi}$ never prevents~${\discandi}$ from winning.
Hence, if the given instance is a {\yesins}, there must be a feasible solution consisting of exactly $\min \{|\Pi_{\discandi}|, \ell\}$ votes.
We reset $\ell:=\min \{|\Pi_{\discandi}|, \ell\}$, and seek a feasible solution with~$\ell$ votes in~$\Pi_{\discandi}$.
Obviously, the final score of~${\discandi}$ is~${\sf{sc}}(\discandi)+\ell$.
If there is a candidate $c\in \mathcal{C}$ such that ${\sf{sc}}(c)\geq {\sf{sc}}(\discandi)+\ell$, the given instance must be a {\noins}.
Assume that this is not the case.
The question is then whether there are~$\ell$ votes in~$\Pi_{\discandi}$ such that for every $c\in B$ at most ${\sf{sc}}(\discandi)+\ell-{\sf{sc}}(c)-1$ of the votes approve~$c$.
This can be solved in {\fpt} time with respect to~$|B|$. To this end, we give an integer linear programming (ILP) formulation with the number of variables being bounded by a function of~$|B|$.
%First, we divide~$\Pi_{\discandi}$ into at most ${|B| \choose r}$ submultisets such that two votes $\pi, \pi'\in \Pi_{\discandi}$ are in the same submultiset if and only if $C(\pi)=C(\pi')$.
%For a submultiset where the set of the top-$r$ candidates in each vote is~$\beta$, we use $\beta$-multiset to denote this submultiset of votes.
We call a vote $\pi\in \Pi_p$ a $\beta$-vote if $\beta=C(\pi)\setminus \{p\}$.
First, we create for each subset $\beta\subseteq B$ an integer variable~$x_{\beta}$ which indicates the number of $\beta$-votes that are included in the solution.
The restrictions are as follows. Let~$n_{\beta}$ be the number of $\beta$-votes in in~$\Pi_p$.
First, for each variable~$x_{\beta}$, we require that $0\leq x_{\beta}\leq n_{\beta}$.
Second, the sum of all variables should be~$\ell$, i.e., $\sum_{\beta\subseteq B} x_{\beta}=\ell$.
Third, for each $c\in B$, it must be that ${\sf{sc}}(c)+\sum_{c\in \beta}x_{\beta}\leq {\sf{sc}}(\discandi)+\ell-1$.
By the result of Lenstra~\shortcite{lenstra83}, this ILP can be solved in {\fpt} time with respect to~$|B|$.
Due to Observation~\ref{obs-r-approval-candidates-consecutive-single-peaked},~$B$ contains at most $k\cdot 2(r-1)$ candidates. The theorem follows.
\end{proof}

Note that the {\fpt}-algorithm in the proof of Theorem~\ref{thm_r_approval_k_axis_fpt} does not need any $k$-axes of the given election.
What important is that when the given election is a $k$-axes single-peaked, the cardinality of the set~$B$ is bounded from above by $k\cdot 2(r-1)$.
The framework in the proof does not apply to CCDV for $r$-approval in $k$-axes elections.
The reason is that any optimal solution of CCDV consists of only votes disapproving the distinguished candidate~${\discandi}$.
Hence, we cannot only confine ourselves to a limited number of candidates.
%The complexity of CCDV for $r$-approval for every $r\geq 3$ in $2$-axes elections remains open.

Now we consider $k$-CP elections. Yang and Guo~\shortcite{DBLP:journals/jcss/YangG17} developed a polynomial-time algorithm for CCAV for $r$-approval in $2$-peaked elections.
As $2$-CP elections are a special case of $2$-peaked elections, their polynomial-time algorithm directly applies to CCAV for $r$-approval in $2$-CP elections.
However, if~$k$ increases just by one, we show that the problem becomes {\nph} even for $r=4$.
%\onlyfull
{Yang and Guo~\shortcite{DBLP:journals/jcss/YangG17} also proved that CCAV for $r$-approval in $3$-peaked elections is {\nph} for every $r\geq 4$.
Their proof is via a reduction from the {\prob{Independent Set on Graphs of Maximum Degree $3$}} problem and, more importantly, the election  constructed in their proof is not a $3$-CP election.}
%\onlyfull
{We use a completely different reduction to show our result. Particularly, our reduction is from the RX3C problem.
%As $k$-CP elections are a special case of $k$-peaked elections, our result is stronger than theirs.
The following lemma is easy to see.

\begin{lemma}
\label{lem-single-peaked-consecutive}
Let~$\lhd$ be a linear order over~$\mathcal{C}$ and let $C\subseteq \mathcal{C}$  be a subset of candidates that are consecutive in~$\lhd$.
Then we can construct a linear order~$\pi$ over~$\mathcal{C}$ such that all candidates in~$C$ are ranked above all candidates not in~$C$,
and~$\pi$ is single-peaked with respect to~$\lhd$.
\end{lemma}
}

\begin{theorem}
\label{thm_ccav_r_approval_3_cp_nph}
CCAV for $r$-approval in $3$-CP elections is {\nph} for every $r\geq 4$.
\end{theorem}

%\onlyfull
{
\begin{proof}
Let $(\xs=\{\xse_1,\dots,\xse_{3\xsize}\}, \xc=\{\xce_1,\dots,\xce_{3\xsize}\})$ be an instance of RX3C\@.
We create a CCAV instance with the following components. Consider first $r=4$.

{\bf{Candidates~$\mathcal{C}$.}} We create in total $24\kappa+5$ candidates.
In particular, for each $\xse_x\in \xs$, we create a set $C(\xse_x)=\{\xse_x^1, \xse_x^2, \xse_x^3, \xse_x^4\}$ of four candidates. Let
\[C_1=\bigcup_{\xse_x\in \xs} C(\xse_x)\] be the set of all these candidates.
Hence, $|C_1|=4\cdot 3\xsize=12\kappa$.
In addition, for each $s=\{\xse_x,\xse_y,\xse_z\}\in \xc$, we first create three candidates $\xse_x(s)$,~$\xse_y(s)$,
and~$\xse_z(s)$ corresponding to~$\xse_x$,~$\xse_y$, and~$\xse_z$, respectively, then we create one candidate~$s'$.
Let~$C_2$ be the set of all these candidates corresponding to all $\xce\in \xc$. Hence, $|C_2|=4\cdot 3\xsize=12\kappa$.
Finally, we create a set~$C_3$ of five candidates denoted by~$\discandi$,~$q_1$,~$q_2$,~$q_3$, and~$q_4$, respectively.
The distinguished candidate is~${\discandi}$.

We now construct the registered and unregistered votes.
For each vote to be created below, we only first specify the approved candidates in the vote, and then we discuss the $2$-axes and use
Lemma~\ref{lem-single-peaked-consecutive} to specify the linear preference of the vote.

{\bf{Registered Votes~$\Pi_{\mathcal{V}}$.}}
First, we create~$5\kappa-1$ votes approving~$q_1$,~$q_2$,~$q_3$, and~$q_4$. Then, for each $s=\{\xse_x, \xse_y, \xse_z\}\in \xc$,
we create $5\kappa-2$ votes approving~$\xse_x(s)$,~$\xse_y(s)$,~$\xse_z(s)$, and~$s'$.
Finally, for each $\xse_x\in \xs$, there are $5\kappa-2$ votes approving~$\xse_x^1$,~$\xse_x^2$,~$\xse_x^3$, and~$\xse_x^4$.

{\bf{Unregistered Votes~$\Pi_{\mathcal{W}}$.}} For each $\xce=\{\xse_x, \xse_y, \xse_z\}\in \xc$, we create four votes as follows:
\begin{itemize}
\item $\pi_{\xce}$ approving~$\xse_x(s)$,~$\xse_y(s)$,~$\xse_z(s)$,~$\discandi$;
\item $\pi_{\xce}^x$ approving~$\xse_x(s)$,~$\xse_x^1$,~$\xse_x^2$,~$\discandi$;
\item $\pi_{\xce}^y$ approving~$\xse_y(s)$,~$\xse_y^1$,~$\xse_y^2$,~$\discandi$; and
\item $\pi_{\xce}^z$ approving~$\xse_z(s)$,~$\xse_z^1$,~$\xse_z^2$,~$\discandi$.
\end{itemize}

Finally, we set $\ell=5\kappa$, i.e., we\onlyfull{ are allowed to} add at most~$5\kappa$ votes.

%\onlyfull
{The above construction clearly takes polynomial time.} Now we discuss the preferences of the above votes. Let
\[\lhd_1=\left(\xse_1^1,\xse_1^2,\xse_1^3,\xse_1^4,\dots,\xse_{3\xsize}^1, \xse_{3\xsize}^2, \xse_{3\xsize}^3, \xse_{3\xsize}^4\right).\]
Let~$\lhd_2$ be an order of~$C_2$ such that for every~$\xce_i$, $i\in [3\xsize-1]$, all candidates created for~$\xce_i$ are ordered before all candidates created for~$\xce_{i+1}$.
The relative order over the candidates created for each $\xce_i=\{\xse_x, \xse_y, \xse_z\}$ can be any liner order 
such that the candidates~$\xse_x(\xce_i)$,~$\xse_y(\xce_i)$, and~$\xse_z(\xce_i)$ are ranked together. 
Finally, let $\lhd_3=\left(p, q_1, q_2, q_3, q_4\right)$. Clearly, for each $i\in \{1,2,3\}$, the approved candidates restricted to~$C_i$ in each vote lie consecutively on~$\lhd_i$.
Then, due to Lemma~\ref{lem-single-peaked-consecutive} we can specify the preferences of the votes in a way so that $(C_i, \Pi_{\mathcal{V}}^{C_i}\cup \Pi_{\mathcal{W}}^{C_i})$
is single-peaked for each $i\in \{1, 2, 3\}$.
%For instance, for a vote $\pi_{\xce}^x$ corresponding to an $\xce\in \xc$ and $\xse_x\in s$,
%a preference of the vote would be  \[\xse_x(s)\succ \xse_x^1\succ \xse_x^2\succ \discandi\succ \lhd_1[\xse_x^2,\xse_{3\xsize}^4], \lhd_1[\xse_{x-1}^4, \xse_1^1], \lhd_2(\xse_x(s), --]}\]

It remains to prove the correctness.

$(\Rightarrow)$ Assume that $S'\subseteq S$ is an exact $3$-set cover of~$\xs$. Consider the election after adding the following~$5\kappa$ votes:
\begin{itemize}
\item All~$2\kappa$ votes~$\pi_{\xce}$ such that $s\not\in \xc'$;
\item For each $s=\{\xse_x, \xse_y, \xse_z\}\in \xc'$, all three votes~$\pi_{\xce}^x$,~$\pi_{\xce}^y$, and~$\pi_{\xce}^z$.
\end{itemize}
As~$S'$ is an exact $3$-set cover, for each $c\in C_2$ at most one of the above added votes approves~$c$.
As a result, each candidate $c\in C_2$ has final score at most $5\kappa-2+1=5\kappa-1$.
Moreover, for each~$\xse_x^i$, $i\in \{1,2\}$ at most one of the above added votes, corresponding to $\xce\in \xc'$ such that $\xse_x\in s$, approves~$\xse_x^i$.
As a result, each candidate in~$C_1$ has final score at most $5\kappa-2+1=5\kappa-1$ too.
As all unregistered votes approve~${\discandi}$ but none of $\{q_1, q_2, q_3, q_4\}$, the final score of~${\discandi}$ is~$5\kappa$ and
the final score of each~$q_1$,~$q_2$,~$q_3$,~$q_4$ is~$5\kappa-1$. In summary,~${\discandi}$ becomes the unique winner in the final election.

$(\Leftarrow)$ Assume that $\Pi_{W}\subseteq \Pi_{\mathcal{W}}$ such that $|\Pi_W|\leq \ell=5\kappa$ and~${\discandi}$ becomes the unique winner after adding all votes in~$\Pi_W$.
As all registered votes disapprove~${\discandi}$ and there are~$5\kappa-1$ registered votes approving~$q_1$, it must be that $|\Pi_W|=5\kappa$.
As each unregistered vote approves~${\discandi}$, the final score of~${\discandi}$ is~$5\kappa$.
Let $\Pi_S=\{\pi_{\xce} \mid \xce\in \xc\}$ and~$\Pi'$ be the multiset of the remaining unregistered votes.
Then,~$\Pi_W$ contains exactly~$2\kappa$ votes in~$\Pi_S$.
The reason is as follows. If~$\Pi_W$ contains less than~$2\kappa$ votes in~$\Pi_S$, then~$\Pi_W$ must contain at least~$3\xsize+1$ votes in~$\Pi'$.
This implies that there are two votes in $\Pi_W\cap \Pi'$ which approve a common candidate $\xse_x^1\in C_1$ for some $\xse_x\in \xs$,
leading~$\xse_x^1$ to have a final score at least $5\kappa-2+2=5\kappa$. This contradicts that~${\discandi}$ is the unique winner.
Moreover, if~$\Pi_W$ contains some vote $\pi_{\xce}\in \Pi_{\xc}$ where $\xce=\{\xse_x, \xse_y, \xse_z\}$,
then none of $\pi_{\xce}^x$,~$\pi_{\xce}^y$,~$\pi_{\xce}^z$ can be included in~$\Pi_W$,
since otherwise due to the construction of the votes, one of~$\xse_x(s)$,~$\xse_y(s)$, and~$\xse_z(s)$ would have a final score at least~$5\kappa$, contradicting that~${\discandi}$
is the unique winner. Hence, if~$\Pi_W$ contains~$t$ votes in~$\Pi_S$, then $|\Pi_W|\leq t+3(3\xsize-t)=9\kappa-2t$. Hence, $t> 2\kappa$ implies $|\Pi_W|< 5\kappa$, a contradiction.
Therefore,~$\Pi_W$ contains exactly~$2\kappa$ votes in~$\Pi_S$. Let $S'=\{s\in \xc \mid \pi_{\xce}\not\in \Pi_W\}$. Due to the above analysis, it holds that $|S'|=3\xsize-2\kappa=\kappa$.
Moreover, for each~$\pi_{\xce}$ where $s=\{\xse_x, \xse_y, \xse_z\}\in \xc'$,
all three votes~$\pi_{\xce}^x$,~$\pi_{\xce}^y$,~$\pi_{\xce}^z$ are in~$\Pi_W$ (otherwise~$\Pi_W$ contains less than~$5\kappa$ votes).
From the fact that each candidate in $C_1\cup C_2$ can be approved by at most one vote in~$\Pi_W$, it follows that~$S'$ is an exact $3$-set cover.

The {\nphns} of CCAV for $r$-approval for every $r\geq 5$ can be obtained from the above reduction by adding some dummy candidates. 
Precisely, for each ${\xse}_x\in \xs$, we make $r-4$ copies of ${\xse}_x^2$ so that they consecutively lie between ${\xse}_x^1$ and ${\xse}_x^2$ in $\lhd_1$.  
For each $\xce\in \xc$, we make $r-4$ copies of~$s'$ and let them lie consecutively with~$s'$ in $\lhd_2$. 
In addition, we create a set~$P$ of $r-4$ candidates which consecutively lie on the left side of~$p$ in~$\lhd_3$.  
Finally, we create a set~$Q$ of~$r-4$ candidates which consecutively lie on the right side of $q_4$ in~$\lhd_3$.
Each vote approves the same candidates as defined above together with $r-4$ certain dummy candidates, who do not have any chance to become a winner by adding at most~$\ell$ votes. 
Precisely, we have the following registered votes.  

\begin{itemize}
\item $5\xsize$ votes approving $\{q_1, q_2, q_3, q_4\}\cup Q$.

\item For each $s=\{\xse_x, \xse_y,\xse_z\}\in \xc$, $5\xsize-2$ votes approving $c_x(s)$, $c_y(s)$, $c_z(s)$, $s'$, and the $r-4$ copies of $s'$.

\item For each $\xse_x\in \xs$, $5\xsize-2$ votes approving $c_x^1$, ${\xse}_x^2$, $\xse_x^3$, $\xse_x^4$, and the $r-4$ copies of $\xse_x^2$.
\end{itemize}

Regarding the unregistered votes, for each $\xce=\{\xse_x, \xse_y, \xse_z\}\in \xc$, we create four votes respectively approving 
\begin{itemize}
\item $\xse_x(s)$, $\xse_y(s)$, $\xse_z(s)$, $p$, and all candidates in~$P$.
\item $\xse_x(s)$, $\xse_x^1$, $\xse_x^2$, and all the $r-4$ copies of $\xse_x^2$.
\item $\xse_y(s)$, $\xse_y^1$, $\xse_z^2$, and all the $r-4$ copies of $\xse_y^2$.
\item $\xse_z(s)$, $\xse_y^1$, $\xse_z^2$, and all the $r-4$ copies of $\xse_z^2$.
\end{itemize}  
The correctness proof is similar.
\onlyfull{To prove the nonunique-winner model, we add one more registered vote of each type.}
\end{proof}
}

Now we turn our attention to CCDV.  
Yang and Guo~\shortcite{Yangaamas14a} proved that CCDV for $r$-approval in $2$-peaked elections is {\nph} even for $r=3$.
We strengthen their result by showing that the problem remains {\nph} even when restricted to elections that are both $2$-axes single-peaked and $2$-CP single-peaked. 
Our reduction is completely different from theirs.
In fact, to establish our result, we resort to a property of $3$-regular bipartite graphs which has not been used in the proof of Yang and Guo~\shortcite{Yangaamas14a}.
The $3$-regular bipartite graph in our reduction comes from the graph-representation of the RX3C problem.
In general, this property says that for every $3$-regular bipartite graph there are two linear orders over the vertices so that
every edge of the graph is between two consecutive vertices in at least one of the two orders.
We believe that this property is of independent interest. Recall that $3$-regular graphs are those whose vertices are all of degree~$3$.

\begin{lemma}
\label{lem-3-regular-bipartite-decomposable}
Let~$G$ be a $3$-regular bipartite graph with vertex set~$\vset{G}$ and edge set~$\eset{G}$.
Then, there are two linear orders~$\lhd_1$ and~$\lhd_2$ over~$\vset{G}$ and a partition $(A_1, A_2)$ of $\eset{G}$ such that for every $i\in \{1,2\} $ and for every edge $\edge{u}{v}\in A_i$,
it holds that~$u$ and~$v$ are consecutive in~$\lhd_i$. 
\end{lemma}

We defer the proof of this lemma to Appendix. Figure~\ref{fig-property-regular-graph} is an illustrating example.

\begin{figure}
\centering
{
\includegraphics[width=0.7\textwidth]{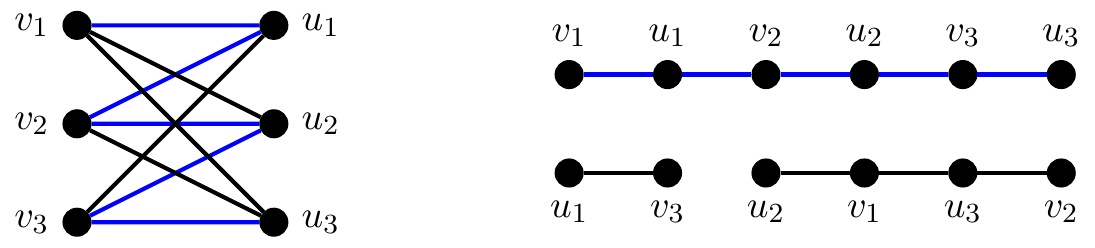}
}
\caption{An illustration of Lemma~\ref{lem-3-regular-bipartite-decomposable}. The left-hand side is a $3$-regular bipartite graph, and the right-hand side showcases 
two linear orders over the vertices so that each edge is connected by two consecutive vertices in at least one of the orders.}
\end{figure}

Now we are ready to unfold the {\nphns} of CCDV for $r$-approval in $2$-axes and $2$-CP single-peaked elections.

\begin{theorem}
\label{thm_ccdv_r_approval_2_cp_np_hard}
For every $r\geq 3$, CCDV for $r$-approval restricted to elections that are both $2$-axes single-peaked and $2$-CP single-peaked is {\nph}.
\end{theorem}

{
\begin{proof}
Let $(\xs=\{\xse_1,\dots,\xse_{3\xsize}\}, \xc=\{\xce_1,\dots,\xce_{3\xsize}\})$ be an instance of RX3C\@. We create a CCDV instance with the following components as follows. 
We first consider $r=3$ and then we discuss how to extend the reduction for any $r>3$. 

{\bf{Candidates~$\mathcal{C}$}.} We create in total $15\kappa+5$ candidates. In particular, for each $\xse_x\in \xs$, we create two candidates~$\xse_x^1$ and~$\xse_x^2$.
In addition, we create five candidates denoted by~$\discandi$,~$q_1$,~$q_2$,~$q_3$, and~$q_4$, where~${\discandi}$ is the distinguished candidate.
Let
\[C_1=\{\xse_x^i \mid i \in \{1,2\}, \xse_x\in \xs\}\cup \{p, q_1, q_2, q_3, q_4\}.\]
Finally, for each $s=\{\xse_x,\xse_y,\xse_z\}\in \xc$, we create three candidates~$\xse_x(s)$,~$\xse_y(s)$, and~$\xse_z(s)$ corresponding to~$\xse_x$,~$\xse_y$, and~$\xse_z$, respectively.
Let~$C_2$ be the set of all candidates corresponding to elements in~$S$. Let $\mathcal{C}=C_1\cup C_2$.

{\bf{Votes~$\Pi_{\mathcal{V}}$}.} \onlyfull{Similar to the proof of Theorem~\ref{thm_ccav_r_approval_3_cp_nph}, w}We only specify here the approved candidates in each created vote,
then after the correctness proof we utilize Lemma~\ref{lem-single-peaked-consecutive} to specify the linear preferences of all votes so that they are both $2$-axes single-peaked 
and $2$-CP single-peaked. 
First, we create one vote approving~$\discandi$,~$q_1$,~$q_2$ and one vote approving~$\discandi$,~$q_3$,~$q_4$. 
In addition, for each $s=\{\xse_x, \xse_y, \xse_z\}\in \xc$, we create four votes as follows: 
\begin{itemize}
\item $\pi_{\xce}$ approving~$\xse_x(s)$,~$\xse_y(s)$,~$\xse_z(s)$;
\item $\pi_{\xce}^x$ approving~$\xse_x(s)$,~$\xse_x^1$,~$\xse_x^2$;
\item $\pi_{\xce}^y$ approving~$\xse_y(s)$,~$\xse_y^1$,~$\xse_y^2$; and
\item $\pi_{\xce}^z$ approving~$\xse_z(s)$,~$\xse_z^1$,~$\xse_z^2$.
\end{itemize}

It is easy to verify that the winning set is $C_1\setminus \{p, q_1, q_2, q_3, q_4\}$.
{Precisely, every winning candidate has score~$3$,~${\discandi}$ has score~$2$, every~$q_i$ where $i\in [4]$ has score~$1$, and every candidate in~$C_2$ has score~$2$.}
Finally, we set $\ell=7\kappa$, i.e., we\onlyfull{ are allowed to} delete at most~$7\kappa$ votes.

The above construction clearly takes polynomial time. 
In the following, we prove the correctness of the reduction.

$(\Rightarrow)$ Assume that $S'\subseteq S$ is an exact $3$-set cover of~$\xs$. Consider the election after deleting the following $7\kappa$ votes:
\begin{itemize}
\item All~$\kappa$ votes~$\pi_{\xce}$ such that $s\in \xc'$;
\item For each $s=\{\xse_x, \xse_y, \xse_z\}\not\in \xc'$, all three votes~$\pi_{\xce}^x$,~$\pi_{\xce}^y$, and~$\pi_{\xce}^z$.
\end{itemize}
Due to the construction and the fact that~$S'$ is an exact 3-set cover,~${\discandi}$\onlyfull{ still} has score~$2$ and every other candidate has score~$1$ after deleting these votes,
implying that~${\discandi}$ becomes the unique winner.

$(\Leftarrow)$ Assume that~$\Pi_{V}$ is a subset of~$\Pi_{\mathcal{V}}$ with minimal cardinality such that $|\Pi_V|\leq \ell=7\kappa$ and~${\discandi}$ becomes the unique winner after
deleting all votes in~$\Pi_V$. Due to the minimality of~$\Pi_V$, no vote in~$\Pi_V$ approves~${\discandi}$.
Hence,~${\discandi}$\onlyfull{ still} has score~$2$ after deleting all votes in~$\Pi_V$. Let $\Pi_S=\{\pi_{\xce} \mid s\in \xc\}$ and $\Pi^U=\{\pi_{\xce}^x \mid s\in \xc, \xse_x\in s\}$.
For every $\pi_{\xce}\not\in \Pi_V\cap \Pi_S$, all three votes~$\pi_{\xce}^x$,~$\pi_{\xce}^y$,~$\pi_{\xce}^z$, where $s=\{\xse_x, \xse_y, \xse_z\}$, must be included in~$\Pi_V$,
since otherwise one of~$\xse_x(s)$,~$\xse_y(s)$, and~$\xse_z(s)$ would have score~$2$ after deleting all votes in~$\Pi_V$. Let $t=|\Pi_V\cap \Pi_S|$.
It follows from the above analysis that $|\Pi_V|\geq t+3(3\xsize-t)=9\kappa-2t$, implying that $t\geq \kappa$.
On the other hand, as there are $6\kappa$ candidates in $C_1\setminus \{p, q_1, q_2, q_3, q_4\}$ and every vote in~$\Pi^U$ approves two of these candidates,
to decrease their scores to at most~$1$, we need to delete at least $9\kappa-3\xsize=6\kappa$ votes, i.e., $|\Pi_V\cap \Pi^U|\geq 6\kappa$.
This directly implies that $t=\kappa$ and $|\Pi_V|=7\kappa$. Let $S'=\{s\in \xc \mid \pi_{\xce}\in \Pi_V\}$. Clearly, $|S'|=t=\kappa$.
Due to the above analysis, for every $\pi_{\xce}\in \Pi_V$, none of~$\pi_{\xce}^x$,~$\pi_{\xce}^y$, and~$\pi_{\xce}^z$ is in~$\Pi_V$,
since otherwise there would be more than~$7\kappa$ votes in~$\Pi_V$. As a result, if there are two~$s, s'\in \xc'$ which contain a
common element $\xse_x\in \xs$, then~$\xse_x^1$ (and~$\xse_x^2$) would have score at least~$2$ after the deletion of all votes in~$\Pi_V$,
contradicting that~${\discandi}$ is the unique winner. So, the $3$-subsets in~$S'$ must be pairwise disjoint, implying that~$S'$ is an exact $3$-set cover.

Finally, we show that the election constructed above is both a $2$-axes election and a $2$-CP election. 
We first how that it is $2$-axes single-peaked. To this end, we show that there exist two axes~$\lhd_1$ and~$\lhd_2$ over~$\mathcal{C}$ such that for every vote constructed above,
the approved candidates in the vote are consecutive in at least one of~$\lhd_1$ and~$\lhd_2$. To this end, we need an auxiliary graph.
Note that the RX3C instance $(\xs, \xc)$ can be represented by a $3$-regular bipartite graph with vertex-partition $(\xs, \xc)$.
In addition, there is an edge between some $\xse\in \xs$ and $\xce\in \xc$ if and only if $\xse\in \xce$.
Due to Lemma~\ref{lem-3-regular-bipartite-decomposable}, there are two linear orders~$\lhd_1'$ and~$\lhd_2'$ over~$\xs\cup \xc$ such that for every edge~$\edge{\xse}{\xce}$
in the graph where $\xse\in \xce\in \xc$ the two vertices~$\xse$ and~$\xce$ are consecutive in one of these two orders.
We first construct a linear order~$\lhd_1^*$ (resp.~$\lhd_2^*$) over $\mathcal{C}$ based on~$\lhd_1'$ (resp.~$\lhd_2'$).
First, we let~$\lhd_1^*$ (resp.~$\lhd_2^*$) be a copy of~$\lhd_1'$ (resp.~$\lhd_2'$) and then we do the following modification.
\begin{itemize}
\item For each $\xse_x\in \xse$ where $x\in [3\xsize]$, we replace~$\xse_x$ with the two candidates~$\xse_x^1$ and~$\xse_x^2$ corresponding to~$\xse_x$ in~$\lhd_1^*$ (resp.~$\lhd_2^*$).
The relative order between~$\xse_x^1$ and~$\xse_x^2$ in~$\lhd_1^*$ (resp.~$\lhd_2^*$) does not matter.
\item For each $s=\{\xse_x, \xse_y, \xse_z\}\in \xc$ where $\{x,y,z\}\subseteq [3\xsize]$,
we replace~$s$ in~$\lhd_1^*$ (resp.~$\lhd_2^*$) with the three candidates~$\xse_x(s)$,~$\xse_y(s)$, and~$\xse_z(s)$ created for the element~$s$.
The relative order among these three candidates are determined as follows. If~$s$
    is not the first element in~$\lhd_1'$ (resp.~$\lhd_2'$), let~$\xse_{i}$, $i\in [3\xsize]$, be the element ordered
    immediately before~$s$ in~$\lhd_1'$ (resp.~$\lhd_2'$), i.e.,~$\xse_{i}$ and~$s$ are consecutive in~$\lhd_1'$ (resp.~$\lhd_2'$) and $\xse_{i}\lhd_1' s$ (resp.~$\xse_{i}\lhd_2' s$).
    If $i\in \{x,y,z\}$, we require that~$\xse_i(s)$ is ordered before everyone in $\{\xse_x(s), \xse_y(s), \xse_z(s)\}\setminus \{\xse_i(s)\}$ so that the three 
    candidates~$\xse_i^1$, $\xse_i^2$, and $c_i(s)$ are consecutive. 
    Symmetrically, if~$s$ is not the last element in~$\lhd_1'$ (resp.~$\lhd_2'$), and~$\xse_{j}$ denotes the element ordered immediately after~$s$ in~$\lhd_1'$ (resp.~$\lhd_2'$)
    we have the following requirement: if $j\in \{x,y,z\}$, we require that~$\xse_j(s)$ is the last one among~$\xse_x(s)$,~$\xse_y(s)$, and~$\xse_z(s)$, so that the 
    three candidates~$\xse_i^1$, $\xse_i^2$, and $c_i(s)$ are consecutive. See Figure~\ref{fig-a} for an illustration. 
    We order~$\xse_x(s)$,~$\xse_y(s)$, and~$\xse_z(s)$ so that the above requirements are fulfilled.
\end{itemize}

\begin{figure}
\centering
{\includegraphics[width=0.65\textwidth]{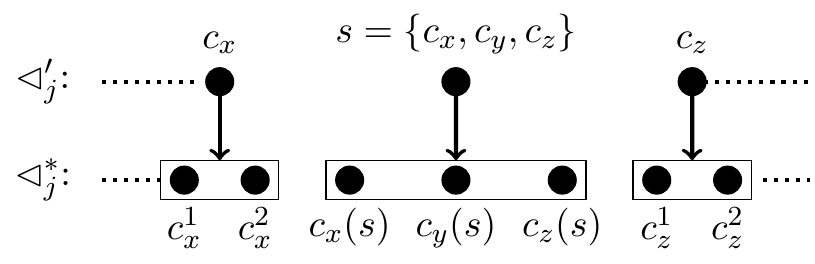}}
\caption{An illustration of the construction of $\lhd_j^*$ from $\lhd_j'$, $j\in [2]$, in the proof of Theorem~\ref{thm_ccdv_r_approval_2_cp_np_hard}.}
\label{fig-a}
\end{figure}

Given the final~$\lhd_1^*$ and~$\lhd_2^*$,  let
    \[\lhd_1=(q_1, q_2, p, q_3, q_4, \lhd_1^*)~\text{and}\] %and
    \[\lhd_2=(q_1, q_2, p, q_3, q_4, \lhd_2^*).\]
    Clearly, the three candidates~$p$,~$q_1$, and~$q_2$ are consecutive in both~$\lhd_1$ and~$\lhd_2$,
    and the three candidates~$p$,~$q_3$, and~$q_4$ are consecutive in both~$\lhd_1$ and~$\lhd_2$ too.
    Due to Lemma~\ref{lem-single-peaked-consecutive}, we can complete the linear order of the vote approving exactly~$p$,~$q_1$, and~$q_2$ (resp.~$p$,~$q_3$, and~$q_4$)
    so that it is single-peaked with respect to~$\lhd_1$ and, moreover,~$p$,~$q_1$, and~$q_2$ (resp.~$p$,~$q_3$, and~$q_4$) are the top-$3$ candidates.
    Let~$s\in \xc$ be a $3$-subset in~$S$. In~$\lhd_1$ and~$\lhd_2$, all the three candidates created for~$s$ are consecutive.
    Let $\xse_x\in s$ be an element in~$s$ and let us consider the vote~$\pi_{\xce}^x$ whose top-$3$ candidates are~$\xse_x(s)$,~$\xse_x^1$, and~$\xse_x^2$.
    Clearly,~$\edge{\xse_x}{s}$ is an edge in the above mentioned $3$-regular graph. 
    Then, due to Lemma~\ref{lem-3-regular-bipartite-decomposable},~$\xse_x$ and~$s$ are consecutive in at least one of the original orders~$\lhd_1'$ and~$\lhd_2'$, say,
    without loss of generality,~$\lhd_1'$. Then due to the definition of~$\lhd_1^*$, the three candidates~$\xse_x^1$,~$\xse_x^2$, and~$\xse_x(s)$ are consecutive.
    Therefore, all the three votes created for~$s$ can be completed into linear-order votes which are single-peaked with respect at least one of~$\lhd_1$ and~$\lhd_2$ too,
    and whose top-$3$ candidates are exactly those that are approved in these votes.
    This completes the proof that the constructed election is a $2$-axes election with~$\lhd_1$ and~$\lhd_2$ being the witness.

Now, we show that the above election is also a $2$-CP election. To this end, it suffices to show that $(C_i, \Pi_{\mathcal{V}}^{C_i})$ is single-peaked for each $i\in [2]$.
Let~$\lhd_1$ be an order of~$C_1$ such that for every~$\xse_x$, $x\in [3\xsize-1]$, the two candidates corresponding to~$\xse_x$ are ordered before the two candidates
corresponding to~$\xse_{x+1}$. Moreover, for each $\xse_x\in \xs$,~$\xse_x^1$ is ordered before~$\xse_x^2$.
Furthermore, the candidates~$p$,~$q_1$,~$q_2$,~$q_3$, and~$q_4$ are ordered after all the other candidates in~$C_1$ and they are ordered as $(q_1, q_2, p, q_3, q_4)$.
Let~$\lhd_2$ be an order of~$C_2$ such that for every $\xce\in \xc$, the three candidates corresponding to~$\xce$ are ordered consecutively
(the relative orders among them do not matter).
Clearly, for each $i\in [2]$ the approved candidates restricted to~$C_i$ in every vote lie consecutively on~$\lhd_i$. 
In this case, we can specify the preferences so that each vote restricted to~$C_i$ is single-peaked with respect to~$\lhd_i$ as follows. 
Let~$\pi$ be a vote. Let~$A_1$ and~$A_2$ be the set of approved candidates in~$\pi$ in~$C_1$  and in~$C_2$, respectively. 
Due to Lemma~\ref{lem-single-peaked-consecutive}, we can specify the preference~$\succ_1$ of~$\pi$ restricted to~$C_1$ (resp.\ $C_2$) so that~$A_1$ (resp.\ $A_2$) are 
ranked consecutively above all the other candidates and the preference is single peaked with respect to~$\lhd_1$ (resp.\ $\lhd_2$). 
Then, we define the preference os~$\pi$ over the whole set of candidates as 
\[A_1\succ A_2\succ C_1\setminus A_1\succ C_2\setminus A_2,\]
where the preferences among candidates in~$A_1$, and among candidates in $C_1\setminus A_1$ are specified by~$\succ_1$, and 
the preferences among candidates in~$A_2$, and among candidates in $C_2\setminus A_2$ are specified by~$\succ_2$.  
This preference is $2$-CP single-peaked with respect to the partition $(C_1, C_2)$ as its restriction to $C_1$ and $C_2$ are exactly $\succ_1$ and $\succ_2$ which are single-peaked 
with respect to $\lhd_1$ and $\lhd_2$, respectively, as discussed above. 

Similar to the proof of Theorem~\ref{thm_ccav_r_approval_3_cp_nph}, to extend the above proof to every $r\geq 4$, we add a number of dummy candidates. 
Precisely, we create $r-3$ copies of~$q_1$ and $r-3$ copies of $q_4$ so that the copies and their originate are consecutive in the respectively axis. 
Let $A(q_1)$ and $A(q_4)$ respectively denote the sets of copies of $q_1$ and $q_4$. 
For each $\xse_x\in \xs$, we create $r-3$ copies of $\xse_x^2$ which consecutively lie with $\xse_x^2$ in the respective linear axis. 
For each $\xce\in \xc$, we create a set $A(\xse)$ of $r-3$ candidates. 
Regarding the votes, we extend the two votes approving $\{p, q_1, q_2\}$ and $\{p, q_3, q_4\}$ respectively to that approving $\{p, q_1, q_2\}\cup A(q_1)$ and $\{p, q_3, q_4\}\cup A(q_4)$. 
For each $\xce\in \xc$, we create four votes approving respectively the following candidates 
\begin{itemize}
\item $\xse_x(s)$,~$\xse_y(s)$,~$\xse_z(s)$ and all candidates in $A(\xse)$;
\item $\xse_x(s)$,~$\xse_x^1$,~$\xse_x^2$ and all copies of $\xse_x^2$;
\item $\xse_y(s)$,~$\xse_y^1$,~$\xse_y^2$ and all copies of $\xse_y^2$;
\item $\xse_z(s)$,~$\xse_z^1$,~$\xse_z^2$ and all copies of $\xse_z^2$.
\end{itemize}
The correctness proof if similar. 
\onlyfull{To prove the nonunique-winner model, we remove the vote approving $\{p, q_3, q_4\}\cup A(q_4)$.}
\end{proof}
}

%%%%%%%%%%%%%
%%%%%%%%%%%%%%
%%%%%%%%%%%%%%%

\section{Condorcet Consistent Voting}
In this section, we study CCAV and CCDV for several Condorcet consistent voting correspondences, i.e., voting correspondences which select exactly the Condorcet winner whenever it exists.
We first show some {\fpt} results for Condorcet.
Our results rely on an {\fpt}-algorithm for the {\prob{Mixed Integer Programming  With Simple Piecewise Linear Transformations}} problem (MIPWSPLT).
This problem is a generalization of integer linear programming (ILP) with the entries of the input matrix being replaced with piecewise linear convex or concave functions.
Bredereck~et~al.~\shortcite{DBLP:journals/corr/abs-1709-02850} recently proved that MIPWSPLT is {\fpt} with respect to the number of variables.
To establish our {\fpt} result, we need only a special case of the MIPWSPLT problem which is defined as follows.

\EP{Integer Programming With Simple Piecewise Linear Transformations (IPWSPLT)}
{A collection of $s\cdot t$ piecewise linear concave functions $\left\{f_{i,j} \setmid i\in [s], j\in [t]\right\}$, and a vector $b\in \mathbb{Z}^s$.}
{Is there a vector~$x=\langle x_1, x_2,\dots, x_t\rangle$ of~$t$ integer variables such that for every $i\in [s]$, it holds that
\begin{equation}
\label{equ-generalized-ilp}
\sum_{j=1}^t f_{i,j}(x_j)\leq b_i?
\end{equation}
}

%In the MIPWSPLT problem studied by in~\cite{DBLP:journals/corr/abs-1709-02850} there are some additional variables which may have real values,
%and there are also piecewise linear convex functions in the formulation. For further details on MIPWSPLT we refer to~\cite{DBLP:journals/corr/abs-1709-02850}.

\begin{lemma}[\cite{DBLP:journals/corr/abs-1709-02850}]
\label{lem-generalized-ILP}
IPWSPLT is solvable in~$\bigos{t^{2.5t+o(t)}}$ time, where~$t$ is the number of variables.
\end{lemma}

Note that the result in~\cite{DBLP:journals/corr/abs-1709-02850} holds for the variant of IPWSPLT where
 the less than sign is replaced with the greater than sign or the equal sign in~\eqref{equ-generalized-ilp}.

For a vote~$\pi$ over a set~$\mathcal{C}$ of candidates and a candidate~$c\in \mathcal{C}$, let~${\sf{Ab}}(\pi,c)$ (resp.\ ${\sf{Be}}(\pi, c)$) be the set of
all candidates ranked above (resp.\ below)~$c$ in~$\pi$, i.e.,
${\sf{Ab}}(\pi,c)=\left\{c'\in \mathcal{C} \setmid \pi(c')<\pi(c)\right\}$ (resp.\ ${\sf{Be}}(\pi, c)=\left\{c'\in \mathcal{C} \setmid \pi(c')>\pi(c)\right\}$).

\begin{theorem}
\label{thm_Condorcet_FPT_k_dimension}
CCAV and CCDV for Condorcet in $k$-axes elections are {\fpt} with respect to~$k$ when $k$-axes are given.
\end{theorem}

{
\begin{proof}
%[CCDV]
We prove the theorem by giving an IPWSPLT formulation for CCDV and CCAV, respectively, with the number of variables being bounded by a function of~$k$.

{Let us consider first CCDV\@. } Let $((\mathcal{C},\Pi_{\mathcal{V}}), \discandi\in \mathcal{C}, \ell)$ be a given
CCDV instance where $(\mathcal{C},\Pi_{\mathcal{V}})$ is a $k$-axes election
with respect to $k$-axes $\lhd_1,\dots,\lhd_{k}$.
We solve the instance as follows. Let $\Pi_{\mathcal{V}_1},\dots,\Pi_{\mathcal{V}_{k}}$ be a partition of~$\Pi_{\mathcal{V}}$ such that
for every $i\in [k]$, all votes in~$\Pi_{\mathcal{V}_i}$ are single-peaked with respect to~$\lhd_i$.
Observe that for each $\pi\in \Pi_{\mathcal{V}_i}$, $i\ [k]$, all candidates ranked above the distinguished candidate~${\discandi}$ lying consecutively on~$\lhd_i$.
Moreover,  either  all of them lie on the left-side of~${\discandi}$ or all of them lie on the right-side of~${\discandi}$ on~$\lhd_i$.
For each~$\Pi_{\mathcal{V}_i}$,~$i\in [k]$, let~$\Pi_{\mathcal{V}_i}^L$ (resp.\ $\Pi_{\mathcal{V}_i}^R$) be the multiset of all votes in~$\Pi_{\mathcal{V}_i}$ where
all candidates ranked above~${\discandi}$ lie on the left-side (resp.\ right-side) of~${\discandi}$ in~$\lhd_i$.
Precisely, for each $i\in [k]$,
\[\Pi_{\mathcal{V}_i}^L=\left\{\pi\in \Pi_{\mathcal{V}_i} \setmid \forall{\left(c\in \mathcal{C}\setminus \{\discandi\}, \pi(c)<\pi(\discandi)\right)}[c\lhd_i p]\right\}~\text{and}\]
\[\Pi_{\mathcal{V}_i}^R=\left\{\pi\in \Pi_{\mathcal{V}_i} \setmid \forall{(c\in \mathcal{C}\setminus \{\discandi\}, \pi(c)<\pi(\discandi))}[p\lhd_i c]\right\}.\]
 For each $i\in [k]$ and $X\in \{L, R\}$, let $t_i^X=|\Pi_{\mathcal{V}_i}^X|$ and $(\pi_{(i,1)}^X, \pi_{(i,2)}^X,\dots,\pi_{(i,t_i^X)}^X)$ be an order over
 $\Pi_{\mathcal{V}_i}^X$ such that $\pi_{(i,x)}^X(\discandi)\geq \pi_{(i,x+1)}^X(\discandi)$ for all $x\in [t_i^X-1]$.
 An observation is that for every~$x$ where $x\in [t_i^X-1]$, ${\sf{Ab}}(\pi_{(i,x+1)}^X, p)\subseteq {\sf{Ab}}(\pi_{(i,x)}^X, p)$, i.e., the candidates ranked above~${\discandi}$
 in~$\pi_{(i,x+1)}^X$ are also ranked above~${\discandi}$ in~$\pi_{(i,x)}^X$.
 This implies that there is an optimal solution such that for each $i\in [k]$ and $X\in \{L, R\}$ such that $t_i^X>0$,
 this solution includes either none of $\Pi_{\mathcal{V}_i}^X$, or it includes all votes~$\pi_{(i,x)}^X$ such that $x\in [y]$ for some positive integer~$y\leq  t_i^X$
 and excludes all the other votes in~$\Pi_{\mathcal{V}_i}^X$.
 Based on the observation, we create an instance of IPWSPLT as follows.
We create in total~$2k$ variables. In particular, for each axis~$\lhd_i$, $i\in [k]$, we create two variables denoted by~$x_i^L$ and~$x_i^R$, where~$x_i^L$ (resp.\ $x_i^R$) indicates
how many votes in~$\Pi_{\mathcal{V}_i}^L$ (resp.\ $\Pi_{\mathcal{V}_i}^R$) are included in the solution. Let $(\xse_1, \xse_2,\dots, \xse_m)$ be any arbitrary but fixed order of
$\mathcal{C}\setminus \{\discandi\}$, where~$m$ is the number of candidates minus one.
For each integer $i\in [k]$, each $X\in \{L, R\}$, and every candidate $c\in \mathcal{C}\setminus \{\discandi\}$,
we define a piecewise linear concave function $f_{i,X,c}: \mathbb{R}_{\geq 0}\rightarrow \mathbb{R}_{\geq 0}$ as follows. First, $f_{i,X,c}(0)=0$.
Second, for each positive integer $x\leq t_i^X$, $f_{i,X,c}(x)$ is the number of votes in $\{\pi_{i,1}^X,\dots,\pi_{i,x}^X\}$ that rank~$c$ above~${\discandi}$ if $t_i^X>0$ and
$f_{i,X,c}(x)=0$ otherwise.
For each integer $x>t_i^X$ we have that $f_{i,X,c}(x)=f_{i,X,c}(t_i^X)$.
Finally, for a real number~$x$ between integers~$y$ and~$y+1$, we have that
\[f_{i,X,c}(x)=f_{i,X,c}(y)+(x-y)\cdot f_{i,X,c}(y+1).\] The restrictions are as follows.
\begin{itemize}
\item For every~$i\in [k]$ and every $X\in \{L, R\}$, we have that $x_i^X\in \mathbb{N}$ and $0\leq x_i^X\leq t_i^X$.
\item Since we seek a feasible solution of size at most~$\ell$, we have that \[\sum_{i\in [k]}(x_{i}^L+x_i^R)\leq \ell.\]
\item To ensure that~${\discandi}$ is the Condorcet winner in the final election, for each candidate $c\in \mathcal{C}\setminus \{\discandi\}$, we have that
\end{itemize}
\[2\left(N\left(c, \discandi\right)-\sum_{\substack{i\in [k]\\ X\in \{L, R\}}}f_{i,X,c}\left(x_i^X\right)\right)<\left |\Pi_{\mathcal{V}}\right|%
-\sum_{\substack{i\in [k]\\ X\in \{L, R\}}} x_i^X.\]
In the above inequality,~$N(c, p)$ is calculated with respect to~$\Pi_{\mathcal{V}}$.
The right side is the number of votes in the final election, and the left side is the double of the number of votes ranking~$c$ above~${\discandi}$ in the final election.
This inequality ensures that~${\discandi}$ beats~$c$ in the final election.
The above programming can be solved in {\fpt} time with respect to~$k$ by the algorithm studied in~\cite{DBLP:journals/corr/abs-1709-02850} (Theorem~2).
%\onlyfull
{
\medskip

Consider now CCAV\@. Let $((\mathcal{C}, \Pi_{\mathcal{V}}), p\in \mathcal{C}, \Pi_{\mathcal{W}}, \ell)$ be a CCAV instance where $(\mathcal{C}, \Pi_{\mathcal{V}}\cup \Pi_{\mathcal{W}})$
is a
$k$-axes election with respect to $k$-axes $\lhd_1, \lhd_2, \dots, \lhd_k$. The algorithm to solve the instance is similar to the above one for CCDV.
Let $(\Pi_{\mathcal{W}_1},\dots,\Pi_{\mathcal{W}_{k}})$ be a partition of~$\Pi_{\mathcal{W}}$ such that for every $i\in [k]$, all votes in~$\Pi_{\mathcal{W}_i}$ are
single-peaked with respect to~$\lhd_i$.
The observation above holds for~$\Pi_{\mathcal{W}}$ as well.
For each $i\in [k]$, let~$\Pi_{\mathcal{W}_i}^L$ (resp.\ $\Pi_{\mathcal{W}_i}^R$) be the multiset of all votes in ~$\Pi_{\mathcal{W}_i}$ (resp.\ $\Pi_{\mathcal{W}_i}$)
where all candidates ranked above~${\discandi}$ lie on the left-side (right-side) of~${\discandi}$ in~$\lhd_i$. Precisely, for each $i\in [k]$,
\[\Pi_{\mathcal{W}_i}^L=\left\{\pi\in \Pi_{\mathcal{W}_i} \setmid c\in \mathcal{C}\setminus \left\{\discandi\right\}, \pi\left(c\right)<\pi\left(\discandi\right),%
c\lhd_i \discandi\right\}~\text{and}\]
\[\Pi_{\mathcal{W}_i}^R=\{\pi\in \Pi_{\mathcal{W}_i} \setmid c\in \mathcal{C}\setminus \{\discandi\}, \pi(c)<\pi(\discandi), \discandi\lhd_i c\}.\]
For each $i\in k$ and each $X\in \{L, R\}$, let $t_i^X=|\Pi_{\mathcal{W}_i}^X|$ and $(\pi_{(i,1)}^X, \pi_{(i,2)}^X,\dots,\pi_{(i,t_i^X)}^X)$ be
an order of all votes in~$\Pi_{\mathcal{V}_i}^X$
such that $\pi_{(i,x)}^X(\discandi)\leq \pi_{(i,y)}^X(\discandi)$ for all integers~$x$ and~$y$ such that $1\leq x<y\leq t_i^X$.
%Similar to the analysis for CCDV, there is an optimal solution which includes all votes $(\pi_{(i,1)}^X, \dots, \pi_{(i,x)}^X)$ for some
% $1\leq x\leq t_i^X$ and excludes all other votes in~$\Pi_{\mathcal{W}_i}^X$.
For each $i\in [k]$, we create two variables denoted by~$x_i^L$ and~$x_i^R$.
In particular,~$x_i^L$ (resp.\ $x_i^R$) indicates how many votes in~$\Pi_{\mathcal{W}_i}^L$ ($\Pi_{\mathcal{W}_i}^R$) are included in the solution.
We define the same piecewise concave functions as in the above algorithm for CCDV. Now, we describe the restrictions.
\begin{itemize}
\item For every $i\in [k]$, we have $0\leq x_i^L\leq t_i^L$ and $0\leq x_i^R\leq t_i^R$.
\item Since we seek a feasible solution of size at most~$\ell$, we have that \[\sum_{i\in [k]}(x_{i}^L+x_i^R)\leq \ell.\]
\item To ensure that~${\discandi}$ is the Condorcet winner in the final election, for each candidate $c\in \mathcal{C}\setminus \{\discandi\}$, we have that
\end{itemize}

\[2\left(N\left(c, \discandi\right)+\sum_{\substack{i\in [k]\\ X\in \{L, R\}}}f_{i,X,c}\left(x_i^X\right)\right)<%
\left |\Pi_{\mathcal{V}}\right |+\sum_{\substack{i\in [k]\\ X\in \{L, R\}}} x_i^X.\]
In the above inequality,~$N(c, p)$ is calculated with respect to~$\Pi_{\mathcal{V}}$.
The above programming can be solved in {\fpt} time with respect to~$k$ by the algorithm studied in~\cite{DBLP:journals/corr/abs-1709-02850} (Theorem~2).
}
\end{proof}
}

Now we consider Condorcet winner in $k$-CP elections.
Yang and Guo~\shortcite{DBLP:journals/jcss/YangG17} proved that CCAV for Condorcet in $3$-peaked elections and CCDV for Condorcet in $4$-peaked elections are {\nph}.
We strengthen their results by showing that CCAV and CCDV for Condorcet are {\nph} in $3$-CP elections, a subclass of $3$-peaked elections.
%We fill the final gaps by showing that CCAV and CCDV for Condorcet in $2$-CP elections are polynomial-time solvable but become {\nph} restricted to $3$-CP elections.
%Our polynomial-time algorithms for $2$-CP elections utilize an intermediate problem which may be of independent interest which, to the
%best of our knowledge, has not been studied to date. For a set $N=\{1,2,\dots, n\}$,
%and interval over~$N$ is a set $[i,j]=\{x\in \mathbb{R}: i\leq x\leq j\}$ such $i\leq j$ and $i, j\in [n]$. An interval~$I$ over~$N$ intersects some $j\in N$ if and only if $j\in I$.

%%%%%%%%%%%%The following theorem is not 100% sure
\onlyfull{
\EP
{Multicolored Exact Intersecting (MEI)}
{A set $N=\{1, 2,\dots, n\}$,~$k$ sets $I_1, I_2, \dots, I_k$ of intervals over~$N$, $k$ non-negative integers $\ell_1, \ell_2, \dots, \ell_k$ summing up to $\ell$,
and a function $b: N\rightarrow \mathbb{N}_{\geq 0}$.}
{Is there a set $I$ of exactly $\ell$ intervals over $N$ such that for every $I_j$ where $j\in [k]$ it holds that $|I\cap I_j|=\ell_j$ and,
moreover, for every $j\in [n]$, there are at most $b(j)$ intervals in $I$ which intersects $j$?}

The MEI problem is a generalization of the {\prob{Independent Set}} problem restricted to interval graphs which is well-known to be polynomial-time solvable~\cite{}.
We show that this generalization can be also solved in polynomial time via a similar algorithm for {\sc{Independent Set}} problem restricted to interval graphs in~\cite{}.

\begin{theorem}
\label{MEI}
The MEI problem is polynomial-time solvable.
\end{theorem}
}

\begin{theorem}
\label{thm-CCAV-CCDV-Condorcet-3-CP-NP-hard}
CCAV and CCDV for Condorcet in $3$-CP elections are {\nph}.
\end{theorem}

{
\begin{proof}%[CCAV]
%We prove only the {\nphns} of CCAV and CCDV for Condorcet in $3$-CP elections.
We reduce the RX3C problem to CCAV and CCDV. Let $(U=\{\xse_1, \dots, \xse_{3\xsize}\}, S=\{\xce_1,\dots, \xce_{3\xsize}\})$ be an RX3C instance.
Consider first CCAV\@.
The components of the CCAV instance are as follows.

{\bf{Candidates~$\mathcal{C}$.}} We create in total $12\kappa$+1 candidates. In particular, for each $\xse_x\in \xs$, we create one candidate~$\xse_x'$.
Let \[C_1=\{\xse_x' \setmid \xse_x\in \xs\}.\]
In addition, for each $s=\{\xse_x,\xse_y,\xse_z\}\in \xc$, we create three candidates~$\xse_x(s)$,~$\xse_y(s)$,~$\xse_z(s)$ corresponding to~$\xse_x$,~$\xse_y$,~$\xse_z$, respectively.
Let~$C_2$ be the set of candidates corresponding to the $3$-subsets in~$S$.
Finally, we create a distinguished candidate~${\discandi}$. Let $C_3=\{\discandi\}$.

Let $\lhd_1=(\xse_1',\dots,\xse_2')$ be the order of~$C_1$ according to the indices of the candidates.
Moreover, let~$\lhd_2$ be any arbitrary order of~$C_2$ such that for each~$s\in \xc$ the three candidates corresponding to~$s$ are ordered consecutively.

{\bf{Registered Votes~$\Pi_{\mathcal{V}}$}.} We create in total $5\kappa-3$ votes, each of which ranks~${\discandi}$ in the last position.
The positions of other candidates in a vote are set in a way so that the vote restricted to~$C_1$ and~$C_2$ is single-peaked with respect to~$\lhd_1$ and~$\lhd_2$, respectively.

{\bf{Unregistered Votes~$\Pi_{\mathcal{W}}$}.} The unregistered votes are created according to~$S$.
In particular, for each $s=\{\xse_x, \xse_y, \xse_z\}\in \xc$, we create four votes~$\pi_{\xce}$,~$\pi_{\xce}^x$,~$\pi_{\xce}^y$, and~$\pi_{\xce}^z$ such that:
\begin{itemize}
\item $\pi_{\xce}\left(\left\{\xse_x(s), \xse_y(s), \xse_z(s)\right\}\right)=\{1,2,3\}$, $\pi_{\xce}(\discandi)=4$;
\item $\pi_{\xce}^x\left(\xse_x(s)\right)=1$, $\pi_{\xce}^x(\xse_x')=2$, $\pi_{\xce}^x(\discandi)=3$;
\item $\pi_{\xce}^y(\xse_y(s))=1$, $\pi_{\xce}^y(\xse_y')=2$, $\pi_{\xce}^y(\discandi)=3$; and
\item $\pi_{\xce}^z(\xse_z(s))=1$, $\pi_{\xce}^z(\xse_z')=2$, $\pi_{\xce}^z(\discandi)=3$.
\end{itemize}
The exact positions of~$\xse_x(s)$,~$\xse_y(s)$,~$\xse_z(s)$ in~$\pi_{\xce}$,
and the positions of the remaining candidates in each of the above four votes are set in a way so that the votes restricted to~$\lhd_1$ and~$\lhd_2$ are single-peaked.
Let $\Pi_S=\{\pi_{\xce} \setmid s\in \xc\}$ and $\Pi^U=\{\pi_{\xce}^x \setmid s\in \xc, \xse_x\in s\}$.

Finally, we set $\ell=5\kappa$, i.e., we are allowed to add at most~$5\kappa$ votes from~$\Pi_{\mathcal{W}}$.
{The above construction clearly takes polynomial time.} It remains to prove the correctness.

$(\Rightarrow)$ Assume that $S'\subseteq S$ is an exact set cover of~$U$. Consider the election after adding the following $5\kappa$ votes:
\begin{enumerate}
\item All~$2\kappa$ votes~$\pi_{\xce}$ such that $s\not\in \xc'$;
\item For each $s=\{\xse_x, \xse_y, \xse_z\}\in \xc'$, all the three votes~$\pi_{\xce}^x$,~$\pi_{\xce}^y$, and~$\pi_{\xce}^z$.
\end{enumerate}
Clearly, the final election has in total $10\kappa-3$ votes. Moreover, for each $c\in C_2$ at most one of the above added~$5\xsize$ votes ranks~$c$ above~${\discandi}$.
As a result, there are at most $(5\kappa-3)+1=5\kappa-2$ votes ranking~$c$ above~${\discandi}$, implying that~${\discandi}$ beats every candidate in~$C_2$ in the final election.
As~$S'$ is an exact set cover, for each~$\xse_x'$, among the~$5\xsize$ added votes only the vote~$\pi_{\xce}$, corresponding to $\xce\in \xc'$ such that $\xse_x\in s$, 
ranks~$\xse_x'$ above~${\discandi}$.
Analogous to the above analysis, we know that~${\discandi}$ beats every candidate in~$C_1$ in the final election. In summary,~${\discandi}$ becomes the Condorcet winner after adding the above~$5\kappa$ votes.

$(\Leftarrow)$ Assume that $\Pi_{W}\subseteq \Pi_{\mathcal{W}}$ such that $|\Pi_W|\leq \ell=5\kappa$ and~${\discandi}$ becomes the Condorcet winner after adding all votes in~$\Pi_W$. 
As~$p$ is not the Condorcet winner with respect to the registered vote constructed above, it holds that $\abs{\Pi_W}\geq 1$. 
Then, we can observe that $|\Pi_W|=5\kappa$ must hold, since otherwise there is at least one candidate which ranked above~${\discandi}$ in at least $(5\kappa-3)+1=5\xsize-2$ 
and hence is not beaten by~${\discandi}$ in the final election. Moreover,~$\Pi_W$ contains exactly~$2\kappa$ votes in~$\Pi_S$. The reason is as follows.
If~$\Pi_W$ contains less than $2\kappa$ in~$\Pi_S$, then~$\Pi_W$ contains more than~$3\xsize$ votes from~$\Pi^U$.
This implies that there are two votes in $\Pi_W\cap \Pi^U$ both of which rank a common candidate $c\in C_2$ above~${\discandi}$,
leading to~$c$ not being beaten by~${\discandi}$ in the final election.
On the other hand, if~$\Pi_W$ contains some vote $\pi_{\xce}\in \Pi_S$ where $s=\{\xse_x, \xse_y, \xse_z\}$,
then none of~$\pi_{\xce}^x$,~$\pi_{\xce}^y$, and~$\pi_{\xce}^z$ can be included in~$\Pi_W$, since otherwise due to the construction of the votes,
one of~$\xse_x(s)$,~$\xse_y(s)$,~$\xse_z(s)$ is not beaten by~${\discandi}$ in the final election.
Hence, if~$\Pi_W$ contains~$t$ votes in~$\Pi_S$, then $|\Pi_W|\leq t+3(3\xsize-t)=9\kappa-2t$, which is strictly smaller than~$5\kappa$ if $t>2\kappa$, a contradiction. 
Let $S'=\{s\in \xc \setmid \pi_{\xce}\not\in \Pi_W\}$. Due to the above analysis, it holds that
\[|S'|=3\xsize-|\Pi_W\cap \Pi_S|=3\xsize-2\kappa=\kappa.\]
Moreover, for each~$\pi_{\xce}$ where $s\in \xc'$ and $s=\{\xse_x, \xse_y, \xse_z\}$, all three votes~$\pi_{\xce}^x$,~$\pi_{\xce}^y$,~$\pi_{\xce}^z$ are in~$\Pi_W$
(otherwise~$\Pi_W$ contains less than~$5\kappa$ votes).
As for each candidate~$c\in C_1$ there can be at most one vote in~$\Pi_W$ ranking~$c$ above~${\discandi}$, it follows that~$S'$ is an exact set cover.
%\onlyfull

{
Consider now the reduction for CCDV for Condorcet in $3$-CP elections.
We first create the same candidates as in the above reduction for CCAV, and then we create one more candidate~$q$ in~$C_3$.
Hence, we have $C_3=\{\discandi, q\}$ now. Let~$\lhd_1$ and~$\lhd_2$ be defined as above. Concerning the votes, we adopt all~$12\kappa$ votes in~$\Pi_S$ constructed above,
with the candidate~$q$ being ranked immediately above~${\discandi}$ (hence, if~$\pi(\discandi)=t$ in a vote~$\pi$ in advance, we have now $\pi(q)=t$ and~$\pi(\discandi)=t+1$).
In addition, we create a multiset of two votes such that~${\discandi}$ and~$q$ are ranked in the $1$st and $2$nd positions, respectively.
Finally, we create a multiset of $5\kappa-1$ votes such that~${\discandi}$ and~$q$ are ranked in the second-last and the last positions, respectively.
Let~$\Pi$ be the multiset of the above~$5\kappa+1$ votes.
The positions of all candidates other than~$\discandi$ and~$q$ in each of~$\Pi$ are set in a way so that this vote restricted to~$\lhd_1$ and~$\lhd_2$ is single-peaked.
In total, we have~$17\kappa+1$ votes. We set~$\ell=7\kappa$, i.e., we are allowed to delete at most~$7\kappa$ votes.
Clearly, the construction can be done in polynomial-time.
Utilizing similar arguments as in the above proof for the correctness of the reduction for CCAV,
we can show that there is an exact set cover of~$U$ if and only if the CCDV instance has a solution of size~$7\kappa$. Precisely, let $S'\subseteq S$ be an exact set cover, then
$\Pi_V=\{\pi_{\xce} \setmid s\in \xc'\}\cup \{\pi_{\xce}^x \setmid s\in \xc\setminus S', \xse_x\in s\}$ is a solution.
One can check that after the deletion of all votes in~$\Pi_V$, for every candidate $c\in \mathcal{C}\setminus \{\discandi\}$,
there are exactly~$5\kappa$ votes ranking~$c$ above~${\discandi}$.
As there remain $17\kappa+1-7\kappa=10\kappa+1$ votes in total,~${\discandi}$ becomes the Condorcet winner.
A significant observation for the proof of the other direction is that any optimal solution of the CCDV instance is disjoint with~$\Pi$, since otherwise~$q$ would beat~${\discandi}$.
Analogous to the above proof for CCAV, we can first show that any solution~$\Pi_V$ contains exactly~$\kappa$ votes in~$\Pi_S$
and~$6\kappa$ votes in~$\Pi^U$, where $\Pi^U=\{\pi_{\xce}^x \setmid s\in \xc, \xse_x\in s\}$.
Then, we can show that $S'=\{s\in \xc \setmid \pi_{\xce}\in \Pi_V\}$ is an exact set cover of~$U$. }
\end{proof}
}

Now we discuss CCAV and CCDV for Copeland$^{\alpha}$, where $0\leq \alpha\leq 1$, and Maximin in $k$-axes elections for small values of~$k$.
In a sharp contrast to the fixed-parameter tractability of CCAV and CCDV for Condorcet in $k$-axes elections, the same problems for both Copeland$^{\alpha}$ and
Maximin are {\nph} even for $k=2$. In particular, Yang and Guo~\shortcite{Yang2014} and Yang~\shortcite{DBLP:phd/dnb/Yang15} established reductions from the X3C problem to CCAV and CCDV for
Copeland$^{\alpha}$, $0\leq \alpha< 1$, in elections with single-peaked width~$2$. It turned out
that the elections constructed in their proofs are $2$-axes single-peaked, as shown in the proof of the following theorem. 
%For Copeland$^{\alpha}$, $0\leq \alpha<1$, we define axes to show that the elections in~\cite{Yang2014} are $2$-axes single-peaked.  
In addition, for CCAV and CCDV for Copeland$^1$ and Maximin, Yang and Guo~\shortcite{Yang2014} and Yang~\cite{DBLP:phd/dnb/Yang15}
proved they are {\nph} in elections with single-peaked~$3$, which are again turned out to be $2$-axes single-peaked. 
%For Copeland$^1$ and Maximin, the elections constructed in the {\nphns} of CCAV and CCDV in~\cite{DBLP:phd/dnb/Yang15} are $2$-axes single peaked too. 
Nevertheless, for these two rules, we provide new reductions because of the following reasons. 
First, compared with the reductions in~\cite{Yang2014,DBLP:phd/dnb/Yang15}, the new reductions are simpler with less candidates, votes, and types of votes. The simplicity is not purely because that we use a reduction from the restricted version of the X3C problem. In fact, some proofs in~\cite{Yang2014,DBLP:phd/dnb/Yang15} heavily rely on the assumption that every element in the universe~$\xs$ occurs in an even number of $3$-subsets in the given collection~$\xc$ which is obviously not fulfilled in any RX3C instance. Second, our reductions for CCAV and CCDV are unified reductions in the sense that they apply to both Copeland$^1$ and Maximin, but in~\cite{Yang2014,DBLP:phd/dnb/Yang15} there are separate reductions for 
Copeland$^1$ and Maximin.

A general explanation of the complexity difference of Condorcet, Maximin, and Copeland$^{\alpha}$, $0\leq \alpha< 1$, 
is that to make the distinguished candidate~${\discandi}$ the Condorcet winner, we need only to focus on
the comparisons between~${\discandi}$ and every other candidate. \onlyfull{Hence, in each vote, only the set of candidates ranking above~${\discandi}$ matters;
how the candidates in $\mathcal{C}\setminus \discandi$ are ranked do no play any role in the algorithm.} In other words, if two votes rank the same set of candidates above~${\discandi}$,
they have the same impact on the solution. However, in Copeland$^{\alpha}$ and Maximin this does not hold.

\begin{theorem}
\label{thm_CCAV_CCDV_Copeland_k_additional_Axis_NP_hard}
CCAV and CCDV for Copeland$^{\alpha}, 0\leq \alpha\leq 1$, and Maximin in $2$-axes elections are {\nph}.
\end{theorem}

%\onlyfull
{
\begin{proof}
In this proof, we first show the {\nphns} of CCAV and CCDV for Copeland$^{\alpha}$, $0\leq \alpha<1$, in $2$-axes elections. This is done by showing that the elections constructed in the {\nphns} reductions of CCAV and CCDV for Copeland$^{\alpha}, 0\leq \alpha\leq 1$, in~\cite{DBLP:phd/dnb/Yang15} (Theorem~3.2) 
are $2$-axes single-peaked. To make the proof as complete as possible, we provide the definitions of the elections but refer the correctness 
proofs to~\cite{DBLP:phd/dnb/Yang15}. 
After this, we derive reductions for CCAV and CCDV in $2$-axes elections which apply to both Copeland$^1$ and Maximin. 
All reductions are from the RX3C problem. Let $(\xs=\{\xse_1,\dots,\xse_{3\xsize}\}, \xc=\{\xce_1,\dots,\xce_{3\xsize}\})$ be an instance of the RX3C problem.

%Let us first discuss CCAV for Copeland$^{\alpha}$, where $0\leq \alpha<1$.  

\bigskip
%%%%%%%%%%%%%% the following reduction for CCAV for Copeland is exactly the same as in the Phd thesis. But the 2-axes should be pointed out here.
{\textbf{CCAV for Copeland$^{\alpha}$, $0\leq \alpha<1$.}}  The constructed election in~\cite{Yang2014} is as follows.  

{\bf{Candidates~$\mathcal{C}$}.} For each $\xse_x\in \xs$ there are two candidates~$\xse_x^L$ and~$\xse_x^R$.
In addition, there are two candidates~${\discandi}$ and~$p'$. The distinguished candidate is~${\discandi}$.

{\bf{Registered Votes~$\Pi_{\mathcal{V}}$}.} There are ${\xsize}-1$ registered votes, each of which has the preference
\[\xse_{3{\xsize}}^L \succ \xse_{3{\xsize}}^R \succ \xse_{3{\xsize}-1}^L\succ \xse_{3{\xsize}-1}^R\succ\cdots\succ p'\succ \discandi.\]
In addition, there is one vote with preference
\[\xse_{3{\xsize}}^R\succ \xse_{3{\xsize}}^L\succ \xse_{3{\xsize}-1}^R\succ \xse_{3{\xsize}-1}^L\succ\cdots\succ {\discandi}\succ {p'}.\]
%The comparisons between the candidates with respect to the registered votes are summarized as follows:
%
%\begin{center}
%\begin{tabular}{|c| c |c |c |c |c |c|}\hline
% & ${\discandi}$ & $p'$ & $\xse_x^L$ & $\xse_x^R$ & $\xse_{x+1}^L$ & $\xse_{x+1}^R$\\ \hline
%
%${\discandi}$ & - & $1$ & \multicolumn{4}{c|}{$0$} \\ \hline
%
%$p'$ & $\kappa-1$ & - & \multicolumn{4}{c|}{$0$} \\ \hline
%
%$\xse_x^L$ & $\kappa$ & $\kappa$ & -  & $\kappa-1$ & \multicolumn{2}{c|}{0}\\  \hline
%
%$\xse_x^R$ & $\kappa$ & $\kappa$ & $1$ & - & \multicolumn{2}{c|}{$0$} \\  \hline
%\end{tabular}
%\end{center}

{\bf{Unregistered Votes:~$\Pi_{\mathcal{W}}$.}} The unregistered votes are created according to~$S$:
for each $s\in \xc$, there is one vote~$\pi_{\xce}$ with preference
\[{\discandi}\succ p'\succ \{\xse_1^L, \xse_1^R\}\succ \{\xse_2^L, \xse_2^R\}\succ \cdots \succ \{\xse_{3\xsize}^L, \xse_{3\xsize}^R\},\]
such that for each $\xse_x\in \xs$ it holds that $\xse_x^L \succ \xse_x^R$ if $\xse_x\in s$ and $\xse_x^R\succ \xse_x^L$ otherwise.

We show that the $(\mathcal{C}, \Pi_{\mathcal{V}}\cup \Pi_{\mathcal{W}})$ is $2$-axes single peaked. 
To this end, let $A=(\xse_1^L,\xse_2^L,\dots,\xse_{3\xsize}^L)$ and $B=(\xse_1^R,\xse_2^R,\dots,\xse_{3\xsize}^R)$.
In addition, let $\lhd_1=(p,A,\overleftarrow{B},p')$ and $\lhd_2=(\overleftarrow{A},p,p',B)$.
It is fairly easy to check that all registered votes are single-peaked with respect to~$\lhd_1$ and all unregistered votes are single-peaked with respect to~$\lhd_2$.

%%%%%%%%%%%%%%%%%%%%%%%% end commen

%Now we discuss CCDV. 
\bigskip
%%%%%%%%%%%%% the following reduction for CCDV for Copeland is also the same as the one in Phd thesis but is simpler because it is reduced from RX3C.
%Again it is important to point out the 2-axes.

{\textbf{CCDV for Copeland$^{\alpha}$, $0\leq \alpha<1$.}} The corresponding election given in~\cite{DBLP:phd/dnb/Yang15} is constructed based on an instance of the 
X3C problem. We replace the X3C instance with an RX3C instance and obtain the following instance. 
The candidate set is the same as the one for CCAV for Copeland$^{\alpha}$, where $0\leq \alpha<1$.
Moreover, let~$\lhd_1$ and~$\lhd_2$ be defined as above. 

{\bf{Votes~$\Pi_{\mathcal{V}}$.}} 
First, for each $s\in \xc$, there is one vote~$\pi_{\xce}$ with the preference
\[\{\xse_{3\xsize}^L, \xse_{3\xsize}^R\}\succ \cdots \succ \{\xse_1^L, \xse_1^R\}\succ \{\discandi\}\succ \{p'\}.\]
Moreover,~$\pi_{\xce}$ prefers~$\xse_x^L$ to~$\xse_x^R$ if and only if $\xse_x\in s$. 
It is fairly easy to see that the above votes are single-peaked with respect to~$\lhd_1$. 
Let $\Pi_{\xc}=\{\pi_{\xce} \setmid s\in \xc\}$. 
In there are~$2\kappa-2$ votes with the preference
\[{\discandi}\succ p'\succ \xse_1^L\succ \xse_1^R\succ \cdots \succ \xse_{3\xsize}^L\succ \xse_{3\xsize}^R,\] and two votes
with preference \[{\discandi}\succ p'\succ \xse_1^R\succ \xse_1^L\succ \cdots \succ \xse_{3\xsize}^R\succ \xse_{3\xsize}^L.\]
One can check that the above $2\kappa$ votes are single-peaked with respect to~$\lhd_2$.
\bigskip

We point out that the above reductions do not apply to Copeland$^1$ and Maximin. In particular, in both reductions, the final score of the distinguished candidate 
is $6\alpha\cdot\kappa+1$ and that of every other candidate is $\alpha \cdot (6\kappa+1)$ after adding (for CCAV) or deleting (CCDV) the votes in a solution. 
Therefore, only when~$0\leq \alpha<1$, the distinguished candidate is the unique winner. 
For Copeland$^{1}$ and Maximin, we derive the following reductions. 
\bigskip

{\textbf{CCAV for Copeland$^{1}$ and Maximin.}} We construct the following instance.  In this reduction, we assume that $\xsize\geq 3$. 

{\bf{Candidates~$\mathcal{C}$}.} We create $9{\xsize}+1$ candidates in total. More specifically, for each $\xse_x\in \xs$ we create a set $C(\xse_x)=\{\xse_x^1, \xse_x^2, \xse_x^3\}$ of three candidates. 
In addition, we create a distinguished candidate~${\discandi}$. Let $\mathcal{C}=\bigcup_{\xse_x\in \xs}C(\xse_x)\cup \{p\}$. 

{\bf{Registered Votes~$\Pi_{\mathcal{V}}$.}} We create in total~$\xsize$ registered votes. Precisely, we first create two votes with the following preference
\[\xse_{3{\xsize}}^1 \succ \xse_{3{\xsize}}^2\succ \xse_{3{\xsize}}^3 \succ \xse_{3{\xsize}-1}^1 \succ \xse_{3{\xsize}-1}^2\succ \xse_{3{\xsize}-1}^3 \succ \cdots \succ {\discandi}.\]
Then, we create $\kappa-2$ votes, each of which has the preference
\[\xse_{3{\xsize}}^2 \succ \xse_{3{\xsize}}^3\succ \xse_{3{\xsize}}^1 \succ \xse_{3{\xsize}-1}^2 \succ \xse_{3{\xsize}-1}^3\succ \xse_{3{\xsize}-1}^1 \succ \cdots \succ {\discandi}.\]
Let $A_1=(\xse_1^1,\xse_2^1,\dots,\xse_{3\xsize}^1)$ and $A_2=(\xse_1^3,\xse_1^2,\xse_2^3,\xse_2^2,\dots,\xse_{3\xsize}^3,\xse_{3\xsize}^2)$. Note that as we assumed $\xsize\geq 3$, 
$\xsize-2$ is a positive integer. 
In addition, let $\lhd_1=(A_1,\overleftarrow{A_2},p)$. It is easy to verify that the above votes are single-peaked with respect to~$\lhd_1$.

{\bf{Unregistered Votes~$\Pi_{\mathcal{W}}$.}} We create~$3\xsize$ unregistered votes based on~$S$.
For each $s\in \xc$, we create a vote~$\pi_{\xce}$ with the preference
\[{\discandi}\succ \{\xse_1^1, \xse_1^2, \xse_1^3\}\succ \cdots \succ \{\xse_{3\xsize}^1, \xse_{3\xsize}^2, \xse_{3\xsize}^3\}.\]
Moreover, for each $\xse_x\in \xs$ we make $\xse_x^2\succ \xse_x^3\succ \xse_x^1$ if $\xse_x\in s$, and $\xse_x^3\succ \xse_x^1\succ \xse_x^2$ otherwise.
Let $A_3=(\xse_1^3,\xse_1^1,\xse_2^3,\xse_2^1,\dots,\xse_{3\xsize}^3,\xse_{3\xsize}^1)$ and $A_4=(\xse_1^2,\xse_2^2,\dots,\xse_{3\xsize}^2)$.
In addition, let $\lhd_2=(\overleftarrow{A_3}, p, A_4)$. It is easy to check that all unregistered votes are single-peaked with respect to~$\lhd_2$.

We complete the construction by setting~$\ell={\xsize}$, i.e., we are allowed to add at most~$\kappa$ votes. 

The reduction can be completed in polynomial time. It remains to prove the correctness. 
We show the correctness for Copeland$^1$ and Maximin together. 
In particular, for the sufficiency direction, we show that adding~$\ell=\xsize$ votes corresponding to an exact set cover of~$\xs$ results in the distinguished candidate~$p$ being the unique 
weak Condorcet winner. It is well-known that both Copeland$^1$ and Maximin are weak Condorcet consistent in the sense that they select exactly the weak Condorcet winners whenever they exist. 
The proof for the necessity is established by showing that the only way to make~$p$ the unique Copeland$^1$/Maximin winner by adding at most~$\ell$ is to make~$p$ the unique weak Condorcet winner. 
%In particular, we show that~$(\xs, \xc)$ has an exact set cover if and only if we can add at most~${\ell}$ unregistered votes to make {$\discandi$} the unique Copeland$^1$ winner.
\smallskip

$(\Rightarrow)$ Let~$S'$ be an exact set cover of~$\xs$ and let $\Pi=\{\pi_{\xce} \setmid s\in \xc'\}$ be the set of the~$\xsize$ unregistered votes corresponding to~$\xc'$. 
In addition, let $\mathcal{E}=(\mathcal{C}, \Pi_{\mathcal{V}}\cup \Pi)$. 
We claim that~$p$ is the unique weak Condorcet winner in~$\mathcal{E}$. 
Due to the above construction, all~$\xsize$ unregistered votes in~$\Pi$ rank~$p$ in the top. As all the~$\xsize$ registered votes rank~$p$ in the last, we know that~$p$ 
is a weak Condorcet winner in~$\mathcal{E}$. It remains to show that there are no other weak Condorcet winners.  
%Moreover, all three candidates corresponding to a~$\xse_x\in \xs$ tie with all three candidates corresponding to a $\xse_y\in \xs\setminus \{\xse_x\}$.
Due to the construction of the unregistered votes and the fact that~$S'$ is an exact set cover,
for every $\xse_x\in \xs$ there is exactly one vote in~$\Pi$ with preference $\xse_x^2\succ \xse_x^3\succ \xse_x^1$ and exactly $\kappa-1$ votes with preference $\xse_x^3\succ \xse_x^1\succ \xse_x^2$. Together with the registered votes, there are more than $2+(\kappa-1)=\xsize+1$
votes preferring~$\xse_x^1$ to~$\xse_x^2$, $\xsize+1$ votes preferring~$\xse_x^2$ to~$\xse_x^3$, and $(\xsize-1)+\xsize=2\xsize-1\geq \xsize+1$ (we assumed $\xsize\geq 3$) votes preferring~$\xse_x^3$ to~$\xse_x^1$, 
implying that every candidate corresponding to $\xse_x\in \xs$ is beaten by at least one candidate in~$\mathcal{E}$.
As this holds for all $\xse_x\in \xs$, we know that no candidate except~${\discandi}$ is a weak Condorcet winner in~$\mathcal{E}$.

$(\Leftarrow)$ Let~$\Pi$ be a solution of the above constructed CCAV instance and let  $\mathcal{E}=(\mathcal{C}, \Pi_{\mathcal{V}}\cup \Pi)$. 
Observe that $|\Pi|=\kappa$,
since otherwise at least one of $\xse_{3\xsize}^1$, $\xse_{3\xsize}^2$, and $\xse_{3\xsize}^3$ is a weak Condorcet winner, implying that~$p$ can be neither the unique Copeland$^1$ winner nor the 
unique Maximin winner in~$\mathcal{E}$. 
Due to the above construction, when~$\abs{\Pi}=\ell$, the number of votes ranking~$p$ in the top equals those ranking~$p$ in the last in~$\mathcal{E}$. 
Therefore,~$p$ ties all the other candidates and hence is a weak Condorcet winner  in~$\mathcal{E}$. 
As Copeland$^1$ and Maximin are weak Condorcet consistent,~$p$ must be the unique Condorcet winner in~$\mathcal{E}$, meaning that every candidate except~$p$ is beaten by at least one candidate. 
Due to the above construction, when $\abs{\Pi}=\ell=\xsize$,  all candidates created for a $\xse_x\in \xs$ ties all candidates created for  another $\xse_y\in \xs\setminus \{\xse_x\}$.
As a result, for every $\xse_x\in \xs$, every~$\xse_x^1$,~$\xse_x^2$, $\xse_x^3$ is beaten by someone in$\{\xse_x^1, \xse_x^2, \xse_x^3\}$.
As all unregistered votes preferring~$\xse_x^3$ to~$\xse_x^1$, and there are $\xsize-2$ registered votes preferring~$\xse_x^3$ to~$\xse_x^1$, there are in total $2\xsize-2\geq \xsize+1$ (recall that we assumed $\xsize\geq 3$) preferring~$\xse_x^3$ to~$\xse_x^1$ in~$\mathcal{E}$. This means that~$\xse_x^3$ beats~$\xse_x^1$ in $\mathcal{E}$. 
Therefore, it must be that~$\xse_x^1$ beats~$\xse_x^2$, and~$\xse_x^2$ beats~$\xse_x^3$ in~$\mathcal{E}$.
As all~$\kappa$ registered votes preferring~$\xse_x^2$ to~$\xse_x^3$,
there is at least one vote $\pi_{\xce}\in \Pi$ preferring~$\xse_x^2$ to~$\xse_x^3$. Due to the construction of the unregistered votes,
vote~$\pi_{\xce}$ prefers $\xse_x^2$ to~$\xse_x^3$ if and only if $\xse_x\in s$. 
As this holds for all $\xse_x\in \xs$, the subcollection corresponding to~$\Pi$, i.e., $\{s \setmid \pi_{\xce}\in \Pi\}$, covers~$\xs$. 
From $\abs{\Pi}=\xsize$, it follows that $\{s \setmid \pi_{\xce}\in \Pi\}$ is an exact set cover of~$\xs$.
%
%The {\nphns} reduction for CCAV-Copeland$^1$-NON can be modified from the above construction by deleting the candidate~$p''$ from the election.
%The correctness argument is similar. %We omit more details in order to get .
\bigskip

{\textbf{CCDV for Copeland$^1$ and Maximin}.}  We construct the following instance. In this reduction, we assume that $\xsize\geq 4$. 
%\onlyfull
{The candidate set is the same as the above one for CCAV for Copeland$^1$ and Maximin. Moreover,~$\lhd_1$ and~$\lhd_2$ are defined the same as above.}

%\onlyijcai{{\bf{Candidates~$\mathcal{C}$}.} We create $9{\xsize}+1$ candidates in total.
%More specifically, for each $\xse_x\in \xs$ we create three candidates~$\xse_x^1$,~$\xse_x^2$, and~$\xse_x^3$. In addition, we create a distinguished candidate~${\discandi}$.
%Let \[\lhd_1=(\xse_{3\xsize}^1,\xse_{3\xsize-1}^1,\dots,\xse_1^1,\xse_1^2,\xse_1^3,\xse_2^2,\xse_2^3,\dots,\xse_{3\xsize}^2,\xse_{3\xsize}^3,\discandi),\]
%%Let $A_3=(\xse_1^1,\xse_1^3,\xse_2^1,\xse_2^3,\dots,\xse_{3\xsize}^1,\xse_{3\xsize}^3)$ and $A_4=(\xse_1^2,\xse_2^2,\dots,\xse_{3\xsize}^2)$. Moreover, let
% \[\lhd_2=(\xse_1^1,\xse_1^3,\xse_2^1,\xse_2^3,\dots,\xse_{3\xsize}^1,\xse_{3\xsize}^3, \discandi, \xse_{3\xsize}^2,\xse_{3\xsize-1}^2,\dots,\xse_1^2).\]}

{\bf{Votes~$\Pi_{\mathcal{V}}$.}} First, for each $s\in \xc$, we create one vote~$\pi_{\xce}$ with the preference
\[\{\xse_1^1, \xse_1^2, \xse_1^3\}\succ \{\xse_2^1, \xse_2^2, \xse_2^3\}\succ\cdots\succ \{\xse_{3\xsize}^1, \xse_{3\xsize}^2, \xse_{3\xsize}^3\}\succ \{\discandi\}.\]
Inside each $\{\xse_x^1, \xse_x^2, \xse_x^3\}$, where $x\in [3\xsize]$, we have $\xse_x^1\succ \xse_x^2\succ \xse_x^3$ if $\xse_x\in s$, and $\xse_x^2\succ \xse_x^3\succ \xse_x^1$ otherwise. 
Let $\Pi_{\xc}=\{\pi_{\xce} \setmid s\in \xc\}$. Clearly, all votes in~$\Pi_{\xc}$ are single-peaked with respect to~$\lhd_1$.
Second, we create a multiset~$\Pi'$ of~$2\kappa$ votes that are single-peaked with respect to~$\lhd_2$. In particular, we create one vote with the preference
\[{\discandi}\succ \xse_{3\xsize}^2\succ \xse_{3\xsize}^3\succ \xse_{3\xsize}^1\succ\cdots\succ \xse_1^2\succ \xse_1^3\succ \xse_1^1,\]
and $2\kappa-1$ votes with the preference \[{\discandi}\succ \xse_{3\xsize}^3\succ \xse_{3\xsize}^1\succ \xse_{3\xsize}^2\succ\cdots\succ \xse_1^3\succ \xse_1^1\succ \xse_1^2.\]
In total, we have~$5\kappa$ votes. It is easy to see that~$\xse_1^2$ is the Condorcet winner assuming $\xsize\geq 4$. 
Finally, let $\ell=\kappa$, i.e., we are allowed to delete at most~$\kappa$ votes. The above construction clearly takes polynomial-time. We show the correctness based on that Copeland$^1$ and 
Maximin are weak Condorcet consistent similar to the above proof for CCAV. 
%Before proceeding further, let's take a look at the comparisons between the candidates.
%
%\begin{center}
%\begin{tabular}{c | c| c| c| c| c}
%
%  &  ${\discandi}$  & $\xse_x'$ & $\xse_{x+1}'$  & $\xse_x''$ & $\xse_{x+1}''$\\ \hline
%
%${\discandi}$ & $-$ &  \multicolumn{4}{c}{$2\kappa$} \\ \hline
%
%$\xse_x'$ & $3\xsize$ & $-$ & $3\xsize$ & $2\kappa+2$ & $5\kappa-1$\\
%\end{tabular}
%\end{center}

$(\Rightarrow)$ Let ${\xc'}\subseteq {\xc}$ be an exact set cover of~$\xs$. Let $\Pi=\{\pi_{\xce} \setmid s\in \xc'\}$ and let $\mathcal{E}=(\mathcal{C}, \Pi_{\mathcal{V}}\setminus \Pi)=(\mathcal{C}, \Pi'\cup (\Pi_{S}\setminus \Pi))$. We argue that~$p$ is the unique weak Condorcet winner in~$\mathcal{E}$. Clearly, $|\Pi|=|S'|=\kappa$. As all the~$2\xsize$ votes in~$\Pi'$ rank~$\discandi$ in the top and all the $3\xsize-\xsize=2\xsize$ votes 
in $\Pi_{S}\setminus \Pi$ rank~$\discandi$ in the last, we know that~${\discandi}$ ties all other candidates in~$\mathcal{E}$ and hence is a weak Condorcet winner.  
%Hence,~${\discandi}$ has Copeland$^1$ score~$9\kappa$ in~$\mathcal{E}$.
%Due to the construction, for two distinct candidates $\xse_x, \xse_y\in \xs$, each of~$\xse_x^1$,~$\xse_x^2$,~$\xse_x^3$ ties with each of~$\xse_y^1$,~$\xse_y^2$,~$\xse_y^3$.
Due to the above construction, for each~$\xse_x\in \xs$ there are exactly $3-1=2$ votes with the preference $\xse_x^1\succ \xse_x^2\succ \xse_x^3$,
and $3\xsize-3-(\kappa-1)=2\kappa-2$ votes with the preference $\xse_x^2\succ \xse_x^3\succ \xse_x^1$ in $\Pi_{\xc}\setminus \Pi$.
Then, as there is exactly one vote with the preference $\xse_x^2\succ \xse_x^3\succ \xse_x^1$
and  $2\kappa-1$ votes with the preference $\xse_x^3\succ \xse_x^1\succ \xse_x^2$ in~$\Pi'$, it holds that~$\xse_x^1$ is beaten by~$\xse_x^3$,~$\xse_x^3$ is beaten by~$\xse_x^2$, and~$\xse_x^2$ is beaten by~$\xse_x^1$ in the election~$\mathcal{E}$, 
implying that none of~$\xse_x^1$,~$\xse_x^2$,~$\xse_x^3$ is a weak Condorcet winner in~$\mathcal{E}$. As this holds for all~$\xse_x\in \xs$, we know that~${\discandi}$ is the unique weak Condorcet winner in~$\mathcal{E}$.

$(\Leftarrow)$ Assume that there is a $\Pi\subseteq \Pi_{\mathcal{V}}$ such that $|\Pi|\leq \kappa$ and~${\discandi}$ becomes the unique Copeland$^1$/Maximin winner
in $\mathcal{E}=(\mathcal{C}, \Pi_{\mathcal{V}}\setminus \Pi)$.
Observe first that it must be that $\Pi\subseteq \Pi_{\xc}$ and $|\Pi|=\kappa$,
since otherwise~$\xse_1^2$ remains as the Condorcet winner. (Recall that in~$\Pi_S$ there are $3\xsize-3$ votes ranking~$c_1^2$ in the top. Assuming $\xsize\geq 4$,~$\xse_1^2$ remains as the Condorcet winner after deleting at most $\ell-1$ votes, or deleting at most~$\ell$ votes in total with at least one of them from~$\Pi'$)
 It follows that~${\discandi}$ ties all the other candidates and hence is a weak Condorcet winner in~$\mathcal{E}$. Because Copeland$^1$ and Maximin are weak Condorcet consistent,~$\discandi$ 
 must be the unique weak Condorcet winner. This means that every candidate except~$\discandi$ is beaten by at least one candidate in~$\mathcal{E}$. 
Due to the above construction, for every two distinct $\xse_x,\xse_y\in \xs$ such that $1\leq x<y\leq 3\xsize$, all the $2\xsize$ votes in $\Pi'$ rank all candidates in~$C({\xse}_y)$ above 
all candidates in~$C(\xse_x)$, and all the $3\xsize-\ell=2\xsize$ votes in $\Pi_S\setminus \Pi$ rank these candidates the other way around. 
Therefore, in the election~$\mathcal{E}$, all candidates in $C(\xse_x)$ tie all candidates in $C(\xse_y)$ for all $y\neq x$. 
%Therefore, for each $\xse_x\in \xs$, every candidate in $\{\xse_x^1, \xse_x^2, \xse_x^3\}$ must be beaten by another candidate in the same set.
Moreover, as all the~$2\kappa$ votes in~$\Pi'$ prefer~$\xse_x^3$ to~$\xse_x^1$, the only candidate which is able to beat $\xse_x^3$ in $\mathcal{E}$ is~$\xse_x^2$.  This further implies  that~$\xse_x^2$ is beaten by~$\xse_x^1$.
As there are exactly $2\kappa-1$ votes in~$\Pi'$ preferring~$\xse_x^1$ to~$\xse_x^2$, there are at least two votes in~$\Pi_{\xc}\setminus \Pi$ preferring~$\xse_x^1$ to~$\xse_x^2$.
Due to the construction, this means that there is at most one vote $\pi_{\xce}\in \Pi$ which prefers~$\xse_x^1$ to~$\xse_x^2$ and $\xse_x\in s$.
As this holds for all $\xse_x\in \xs$ and $|\Pi|=\kappa$, the subcollection $\{s\in \xc \setmid \pi_{\xce}\in \Pi\}$ is an exact set cover of~$\xs$.
\end{proof}
}

Note that the {\nphns} of CCAV and CCDV for Copeland$^{\alpha}$, $0\leq \alpha<1$ in elections with single-peaked width~$2$, established by Yang and Guo~\shortcite{Yangaamas14a},
implies the {\nphns} of the same problems in $2$-CP elections because any election with single-peaked width~$k$ is a $k'$-CP election for some $k'\leq k$~\cite{Erdelyi2017}.

For Copeland$^1$ and Maximin in elections with single-peaked width~$2$,
Yang and Guo~\shortcite{Yangaamas14a} proved that CCAV and CCDV are polynomial-time solvable. Our results stand in contrast to theirs.

\begin{theorem}
\label{thm-CCAV-CCDV-Copeland-1-maximin-2-cp-np-hard}
CCAV and CCDV for Copeland$^1$ and Maximin in $2$-CP elections are {\nph}.
\end{theorem}

%\onlyfull
{
\begin{proof}
We prove the theorem by reductions from the RX3C problem. Let $(\xs,\xc)$ be a given RX3C instance where $\abs{\xs}=\abs{\xc}=3\xsize$. 
Without loss of generality, we assume that $\xsize\geq 3$.
Let $(\xse_1, \xse_2, \dots, \xse_{3\xsize})$ be an arbitrary but fixed order of~$\xs$. 
Similar to Theorem~\ref{thm_CCAV_CCDV_Copeland_k_additional_Axis_NP_hard}, our proofs are based on that Copeland$^1$ and Maximin are weak Condorcet consistent. 
\bigskip

{\textbf{CCAV for Copeland$^1$ and Maximin}.} We create an instance as follows.

{\textbf{Candidates $\mathcal{C}$.}} For each $\xse_x\in \xs$, $x\in [3\xsize]$, we create a set $C_x=\{\xse_x^1,\xse_x^2,\xse_x^3, \xse_x^4\}$ of four candidates.
%Let $C_x^{1,3}=\{{\xse}_x^1, {\xse}_x^3\}$, $C_x^{2,4}=\{{\xse}_x^2, {\xse}_x^4\}$, and
 In addition to these candidates, we create a candidate~$\discandi$ which is the distinguished candidate. In total, we have $12\xsize+1$ candidates.

In the following, we create the votes so that they are $2$-CP single-peaked with respect to the following axis. Particularly, let
\[\lhd_1=(p, {\xse}_1^1, {\xse}_1^3, {\xse}_2^1, {\xse}_2^3, \dots, {\xse}_{3\xsize}^1, {\xse}_{3\xsize}^3)\]
and
\[\lhd_2=({\xse}_1^2, {\xse}_1^4, {\xse}_2^2, {\xse}_2^4, \dots, {\xse}_{3\xsize}^1, {\xse}_{3\xsize}^4).\]
Let $\lhd=(\lhd_1, \lhd_2)$. Clearly, all created candidates are in the linear order~$\lhd$ which is served as the axis.

{\textbf{Registered Votes $\Pi_{\mathcal{V}}$.}} First, we create $\xsize+1$ votes with the preference 
\[C_{3\xsize}\succ C_{3\xsize-1}\succ \cdots\succ C_{1}\succ p.\]
Among these votes, we specify the preferences inside all~$C_x$, $x\in [3\xsize]$, such that the following requirements are fulfilled (number of votes: preferences)
\medskip

\begin{tabular}{rl}
$\xsize-2$: & ${\xse}_x^3\succ {\xse}_x^1\succ {\xse}_x^4\succ {\xse}_x^2$\\
$1$: & ${\xse}_x^4\succ {\xse}_x^2\succ {\xse}_x^3\succ {\xse}_x^1$\\
$2$: & ${\xse}_x^3\succ {\xse}_x^4\succ {\xse}_x^1\succ {\xse}_x^2$\\
\end{tabular}
\medskip

Additionally, we create one vote with the preference
\[p\succ C_1\succ C_2\succ\cdots\succ C_{3\xsize}\] so that inside each~$C_x$, $x\in [3\xsize]$, it holds that
\[{\xse}_x^1\succ {\xse}_x^2\succ {\xse}_x^3\succ {\xse}_x^4.\]

The pairwise comparisons among candidates in each~$C_x$ with respect to the registered votes are shown in Figure~\ref{fig-comparison}.
\begin{figure}
\begin{center}
\includegraphics[width=2cm]{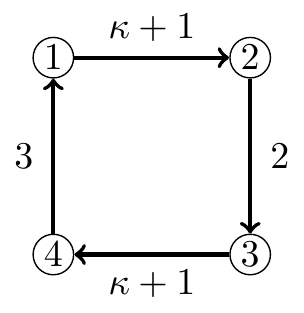}
\end{center}
\caption{The pairwise comparisons among candidates in each~$C_x$ with respect to the registered votes in the proof of CCAV for Copeland$^1$ in
Theorem~\ref{thm-CCAV-CCDV-Copeland-1-maximin-2-cp-np-hard}.
The number beside an arc from a candidate to another candidate is the number of registered votes ranking the former above the latter.
The node with number~$j$ inside presents the candidate ${\xse}_x^j$.
The comparison between ${\xse}_x^1$ and ${\xse}_x^3$, and the comparison ${\xse}_x^2$ and ${\xse}_x^4$ are not given since their comparisons can be easily analyzed base on the single-peakedness in
the correctness proof given below.}
\label{fig-comparison}
\end{figure}
\smallskip

{\textbf{Unregistered Votes $\Pi_{\mathcal{W}}$.}} The unregistered votes are created according to~$\xc$. Particularly, for each $\xce\in \xc$, we create one vote~$\pi_{\xce}$ with the preference
\[p\succ C_1\succ C_2\succ \cdots \succ C_{3\xsize}.\]
In side each~$C_x$ where $x\in [3\xsize]$, we set

\begin{tabular}{rl}
${\xse}_x^1\succ {\xse}_x^2\succ {\xse}_x^3\succ {\xse}_x^4$ & if ${\xse}_x\in \xse$, and\\
${\xse}_x^2\succ {\xse}_x^4\succ {\xse}_x^1\succ {\xse}_x^3$ & if ${\xse}_x\notin \xse$.\\
\end{tabular}

One can check all the above created votes are single peaked with respect to each of~$\lhd_1$ and~$\lhd_2$, and hence the election is a $2$-CP election.

Finally, we set $\ell=\xsize$, i.e., we are allowed to add at most~$\xsize$ unregistered votes.

The construction clearly can be done in polynomial time. In the following, we prove the correctness of the reduction.

$(\Rightarrow)$ Assume that there is an exact set cover $\xc'\subseteq \xc$ of~$\xs$. Consider the new election after adding the~$\xsize$
votes corresponding to~$\xc'$. In the new election, there are exactly $2\xsize+2$ votes among which $\xsize+1$ have the preference
\[C_{3\xsize}\succ \cdots \succ C_1\succ p\] and $\xsize+1$ votes with the preference
\[p\succ C_1\succ \cdots\succ C_{3\xsize}.\]
Therefore, in the new election,~$p$ is a weak Condorcet winner. 
%Moreover, restricting to both $\lhd_1$ and $\lhd_2$, the number of votes whose peaks are at the left-most and right-most candidates are the same.
%This implies that in the new election, ${\xse}_x^1$ ties ${\xse}_x^3$, and ${\xse}_x^2$ ties ${\xse}_x^4$, for all $x\in [3\xsize]$.
Due to the above construction of the unregistered votes, we know that inside each~$C_x$, $x\in [3\xsize]$,
exactly one unregistered vote with the preference ${\xse}_x^1\succ {\xse}_x^2\succ {\xse}_x^3\succ {\xse}_x^4$ and exactly $\xsize-1$ unregistered votes
with the preference ${\xse}_x^2\succ {\xse}_x^4\succ {\xse}_x^1\succ {\xse}_x^3$ are added, which results in total one vote with ${\xse}_x^1\succ {\xse}_x^2$,~$\xsize$
many with ${\xse}_x^2\succ {\xse}_x^3$, one with ${\xse}_x^3\succ {\xse}_x^4$, and $\xsize-1$ many with ${\xse}_x^4\succ {\xse}_x^1$ in the added votes.
With the help of Figure~\ref{fig-comparison}, we can easily check that in the
new election~${\xse}_x^1$ beats~${\xse}_x^2$,~${\xse}_x^2$ beats~${\xse}_x^3$,~${\xse}_x^3$ beats~${\xse}_x^4$, and ~${\xse}_x^4$ beats~${\xse}_x^1$, 
meaning that none of~$C_x$ is a weak Condorcet winner in the new election. Therefore,~$p$ is the unique weak Condorcet winner.

$(\Leftarrow)$ Assume that we can add at most $\ell=\xsize$ unregistered votes so that~$p$ becomes the unique Copeland$^1$/Maximin winner.
Observe that we have to add exactly~$\ell$ votes since otherwise someone in~$C_1$ will be the winner as a majority of votes rank~$C_x$ in
the top in this case. Analogous to the analysis in the above direction, this implies that~$p$ ties all the other candidates and hence is a weak Condorcet winner 
in the final election. As Copeland$^1$ and Maximin are weak Condorcet winner, everyone else must be beaten by at least one candidate. 
Note that in the new election, for different $x,y\in [3\xsize]$, candidates in~$C_x$ tie all candidates in~$C_y$. Moreover, for all $x\in [3\xsize]$, the candidate~$\xse_x^1$ ties the candidate~$\xse_x^3$, and $\xse_x^2$ ties $\xse_x^4$. 
As all unregistered votes and multiple registered votes prefer ${\xse}_x^2$ to ${\xse}_x^3$ for all $x\in [3\xsize]$, ${\xse}_x^2$ beats ${\xse}_x^3$
in the new election. 
Therefore, the only candidate which is  able to beats ${\xse}_x^2$ in the new election is~${\xse}_x^1$.
As there are $\xsize+1$ registered votes preferring ${\xse}_x^1$ to ${\xse}_x^2$, we know that there is at least one vote $\pi_{\xce}$ 
which prefers ${\xse}_x^1$ to ${\xse}_x^2$ and is added in the election.
Due to the above construction, a vote $\pi_{\xce}$ prefers~${\xse}_x^1$ to~${\xse}_x^2$ if and only if ${\xse}_x\in \xce$.
As this holds for all ${\xse}_x\in \xs$, the subcollection corresponding to the added votes must be an set cover of~$\xs$.
As we add exactly~$\ell=\xsize$ votes, it must be an exact set cover.

%% the above proof also applies to the unique winner of weak Condorcet.
%\bigskip

%One can check that the above reduction applies to Maximin too.
%In particular, in order to make~$p$ the unique Maximin winner, we have to add exactly~$\ell=\xsize$ votes which makes~$p$ being tied with all the other candidates.
%Hence, everyone ${\xse}_x^j$ where $x\in [3\xsize]$ and $j\in [4]$  must be beaten by at least one candidate,
%which, due to the single-peakedness, must be from the set~$C(\xse_x)$ as discussed above.
\bigskip

{\textbf{CCDV for Copeland$^1$/Maximin}} The reduction here is similar to the above reduction. The candidates are exactly the same as in the above reduction.
In addition, we adopt all the registered and unregistered votes created in the above reduction but in each of them we reverse the position of the candidate~$\discandi$.
In particular, we have $\xsize+1$ votes with
the preference
\[p\succ C_{3\xsize}\succ C_{3\xsize-1}\succ \cdots \succ C_1\]
where the preferences inside each~$C_x$ are the same as in the above reduction.
In addition, we create one vote~$\pi$ with the preference 

\[C_1\succ C_2\succ \cdots\succ C_{3\xsize}\succ p\]
with the preference inside each~$C_x$, $x\in [3\xsize]$, being ${\xse}_x^1\succ {\xse}_x^2\succ {\xse}_x^3\succ {\xse}_x^4$, as in the above reduction.
Then, for each $\xse\in \xs$, we create a vote~$\pi_{\xse}$ with the preference
\[C_1\succ C_2\succ \cdots \succ C_{3\xsize}\succ p\]
with the preference inside each~$C_x$, $x\in [3\xsize]$, being the same as in the above reduction. Let $\Pi_{\xc}=\{\pi_{\xce}\setmid \xce\in \xc\}$ be the set of the votes corresponding to~$\xc$. 
Let $\Pi_{\mathcal{V}}$ denote the multiset of all votes in this reduction. 
These votes are $2$-axes single-peaked
with respect to a ``locally reversed''~$\lhd_1$ defined above.
Precisely, the new axis is
\[({\xse}_1^1, {\xse}_1^3, {\xse}_2^1,{\xse}_2^3,\dots, {\xse}_{3\xsize}^1, {\xse}_{3\xsize}^3, p, {\xse}_1^2, {\xse}_1^4, {\xse}_2^2, {\xse}_2^4, \dots, {\xse}_{3\xsize}^2, {\xse}_{3\xsize}^4).\]

Finally, we set $\ell=2\xsize$. The reduction can be done in polynomial time. Now we show the correctness.

$(\Rightarrow)$ If there is an exact set cover $\xc'\subseteq \xc$, we delete all the~$2\xsize$ votes corresponding to~$\xc\setminus \xc'$. After deleting
these votes, it holds that~$p$ ties all candidates, and everyone in some~$C_x$ ties everyone in some~$C_y$ such that $x\neq y$. As the preferences inside each~$C_x$
remain unchanged in all votes, for all~$C_x$, $x\in [3\xsize]$ in the new election, everyone in~$C_x$ is beaten by someone in the same set~$C_x$. Therefore,~$p$
is the unique winner in the new election.

$(\Leftarrow)$ Assume that we can remove a multiset~$\Pi$ of at most $\ell=2\xsize$ votes so that~$p$ becomes the unique Copeland$^1$/Maximin winner in $\mathcal{E}=(\mathcal{C}, \Pi_{\mathcal{V}}\setminus \Pi)$. 
As we have exactly $\xsize+1$ votes where~$p$ is ranked in the last position,
and we have in total $4\xsize+2$ votes, it must be that exactly~$\ell=2\xsize$ votes with~$p$ being ranked in the top are deleted (Otherwise,~$p$ is beaten by all candidates).
Recall that only votes in $\Pi_{\xc}\cup \{\pi\}$ rank~$p$ in the top and hence it holds that $\Pi\subseteq \Pi_{\xc}\cup \{\pi\}$. We first claim that the vote~$\pi$ is not deleted, i.e., $\pi\not\in \Pi$. 
To see this, check that for every $x\in [3\xsize]$, all votes in $\Pi_{\xc}\cup \{\pi\}$ prefer ${\xse}_x^2$ to ${\xse}_x^3$,
and in addition to these votes there is one more vote ranking ${\xse}_x^2$ above ${\xse}_x^3$.
This implies that in the election $\mathcal{E}$, there are in total $(3\xsize+1)+1-2\xsize=\xsize+2$ votes ranking ${\xse}_x^2$ above ${\xse}_x^3$.
This means that ${\xse}_x^2$ beats ${\xse}_x^3$ in~$\mathcal{E}$. Moreover, similar to the above analysis, we know that after deleting $2\xsize$ votes ranking~$\discandi$ in the top,~$\discandi$ ties all the other candidates, all candidates from different $C_x$ and $C_y$ where $x,y\in [3\xsize]$ are tied, and for all $x\in [3\xsize]$ the candidate $\xse_x^1$ ties $\xse_x^3$, and $\xse_x^2$ ties $\xse_x^4$. 
Given that Copeland$^1$ and Maximin are weak Condorcet consistent, everyone except~$\discandi$ is beaten by at least one candidate. Due to the above analysis, for each $x\in [3\xsize]$, the only candidate which beats~${\xse}_x^2$ must be~${\xse}_x^1$ for all $x\in [3\xsize]$.
Assume for the sake of contradiction that the vote~$\pi\in \Pi$. Let~$\pi_{\xce}\in \Pi_{\xc}\cap \Pi$ be any arbitrary vote in~$\Pi_{\xc}$ that is deleted too. Such a vote must exist since $2\xsize-1>1$. 
Let~$\xse_x$ be any element in~$\xce$. Due to the definition of $\pi_{\xce}$,~$\pi_{\xce}$ ranks~${\xse}_x^1$ above~${\xse}_x^2$. That is, there is at least one $\xse_x\in \xs$ such that
at least two votes ranking~${\xse}_x^1$ above~${\xse}_x^2$ are deleted, resulting in at most $\xsize+1$ votes ranking ${\xse}_x^1$ above ${\xse}_x^2$ in~$\mathcal{E}$,
implying that~${\xse}_x^2$  is not beaten by~${\xse}_x^1$, a contradiction. 
The claim is proved. Given the claim, we know that $\Pi\subseteq \Pi_{\xc}$. As $\xse_x^1$ beats $\xse_x^2$ where $x\in [3\xsize]$ in~$\mathcal{E}$, 
at most two votes with the preference ${\xse}_x^1\succ {\xse}_x^2$ can be included in~$\Pi$, leaving at least one vote with the preference ${\xse}_x^1\succ {\xse}_x^2$ in $\Pi_S\setminus \Pi$. Let $\xc'=\{\xce\in \xc \setmid \pi_{\xce}\in \Pi_S\setminus \Pi\}$ be the subcollection corresponding to the votes in~$\Pi_S$ that are not deleted. Clearly, $\abs{\xc'}=\xsize$. Because for every $\xce\in \xc$ and every $\xse_x\in \xs$, the vote $\pi_{\xce}$ ranks~${\xse}_x^1$ above~${\xse}_x^2$ if and only if $\xse_x\in \xce$,~$\xc'$ covers~$\xs$. As~$\abs{\xc'}=\xsize$,~$\xc'$ is an exact set cover of~$\xs$. 
%%% the above reduction applies to weak Condorcet winner too. 
\end{proof}
}

%%%%%%%%%%%%%%%%%%%%%%%%%%%%%%%%%%%%%
%%%%%%%%%%%%%%%%%%%%%%%%%%%%%%%%%%%%%
%%%%%%%%%%%%%%%%%%%%%%%%%%%%%%%%%%%%%
\section{Recognition of Nearly Single-Peakedness}
It is known that determining whether an election is single-peaked ($1$-axis) is polynomial-time solvable~\cite{Bartholdi1986T,Doignon:1994:PTA:182528.182531,DBLP:conf/ecai/EscoffierLO08}.
For every $k\geq 3$, Erd\'{e}lyi, Lackner, and Pfandler~\shortcite{Erdelyi2017} proved that determining whether an election is a $k$-axes election is {\nph}.
We complement these results by showing that determining whether an election is a $2$-axes election is polynomial-time solvable,
completely filling the complexity gap of the problem with respect to~$k$.
To this end, we reduce the problem to the {\prob{$2$-Satisfiability}} problem (2SAT).

\EP
{{\prob{$2$-Satisfiability}}}
{A CNF conjunction formula with such that each clause consists of exactly two literals.}
{Is there a truth-assignment which satisfies all clauses?}

It is well-known that the {\prob{$2$}-Satisfiability} problem can be solved in linear time in the number of
clauses~\cite{DBLP:journals/siamcomp/EvenIS76,DBLP:journals/ipl/AspvallPT79,DBLP:conf/dimacs/GuPFW96}.

Single-peaked elections have a nice characterization~\cite{DBLP:journals/scw/BallesterH11} which is useful in our study.
%Let~$(\mathcal{C}, \Pi_{\mathcal{V}})$ be an election.

\begin{definition}[Worst-diverse structure (WD-structure)] A WD-structure $(\pi_x, \pi_y, \pi_z, a,b,c)$ in and election $(\mathcal{C}, \Pi_{\mathcal{V}})$ is a $6$-tuple such that
\begin{itemize}
\item $\pi_x,\pi_y,\pi_z\in \Pi_{\mathcal{V}}$, $a,b,c\in \mathcal{C}$;
\item $\pi_x(a)>\max\{\pi_x(b),\pi_x(c)\}$;
\item $\pi_y(b)>\max\{\pi_y(a),\pi_y(c)\}$; and
\item $\pi_z(c)>\max\{\pi_z(a),\pi_z(b)\}$.
\end{itemize}
\end{definition}

Three votes {\it{WD-conflict}} if there are three candidates forming a WD-structure with them.

\begin{definition}[$\alpha$-structure] An $\alpha$-structure $(\pi_x, \pi_y, a,b,c, d)$ in an election $(\mathcal{C}, \Pi_{\mathcal{V}})$ is a $6$-tuple such that
\begin{itemize}
\item $\pi_x,\pi_y\in \Pi_{\mathcal{V}}$, $a,b,c,d \in \mathcal{C}$;
\item $\pi_x(a)<\pi_x(b)<\pi_x(c)$, $\pi_x(d)<\pi_x(b)$; and
\item $\pi_y(c)<\pi_y(b)<\pi_y(a)$, $\pi_y(d)<\pi_y(b)$.
\end{itemize}
\end{definition}

Two votes {\it{$\alpha$-conflict}} if there are four candidates forming an $\alpha$-structure with them.
The following lemma was studied by Ballester and Haeringer~\shortcite{DBLP:journals/scw/BallesterH11}.

\begin{lemma}[\cite{DBLP:journals/scw/BallesterH11}]
\label{lem-single-peaked-charcteristics}
An election $(\mathcal{C}, \Pi_{\mathcal{V}})$ is single-peaked if and only if there are no WD-structure and $\alpha$-structure in $(\mathcal{C},\Pi_{\mathcal{V}})$.
\end{lemma}

Armed with the above lemma, we are able to develop the polynomial-time algorithm for recognizing $2$-axes elections.

\begin{theorem}
Determining whether an election is a $2$-axes election is polynomial-time solvable.
In particular, it can be solved in $\bigo{n^2\cdot m^3\cdot (n+m)}$ time, where~$m$ and~$n$ denotes the number of candidates and the number of votes in the given election.
\end{theorem}

\begin{proof}
Let $(\mathcal{C}, \Pi_{\mathcal{V}})$ be an election. Let $m=\abs{\mathcal{C}}$ and $n=\abs{\Pi_{\mathcal{V}}}$ denote the number of candidates and the number of votes respectively.
The problem is equivalent to seeking a partition $(\Pi_{T}, \Pi_{F})$ of $\Pi_{\mathcal{V}}$ so that both $(\mathcal{C}, \Pi_{T})$ and $(\mathcal{C}, \Pi_{F})$ are single-peaked.
We reduce the problem to the 2SAT problem as follows.
For each vote $\pi\in \Pi_{\mathcal{V}}$, we create a variable~$x(\pi)$.
Hence, a partition $(\Pi_{T}, \Pi_{F})$ of~$\Pi_{\mathcal{V}}$ corresponds to a truth assignment of the variables, and vice versa:
variables corresponding to votes in~$\Pi_{T}$ are assigned {\sf{true}} and variables corresponding to votes in~$\Pi_{F}$ are assigned {\sf{false}}.

If there are no WD-structures in $(\mathcal{C}, \Pi_{\mathcal{V}})$, we create the clauses as follows. For every two votes~$\pi$ and~$\pi'$ which $\alpha$-conflict,
we create two clauses $\left(x(\pi), x(\pi')\right)$ and $\left(\overline{x(\pi)}, \overline{x(\pi')}\right)$. In order to satisfy both clauses,~$x(\pi)$ and~$x(\pi')$
must be assigned differently, and thus~$\pi$ and~$\pi'$ are partitioned into different sets. From Lemma~\ref{lem-single-peaked-charcteristics},
there is a truth assignment satisfying all clauses if and only if $(\mathcal{C}, \Pi_{\mathcal{V}})$ is a $2$-axes election.
Assume now that there are WD-structures in $(\mathcal{C}, \Pi_{\mathcal{V}})$. Let~$a$,~$b$, and~$c$ be three candidates in a WD-structure.
Let $(\Pi_a, \Pi_b, \Pi_c)$  be a partition of~$\Pi_{\mathcal{V}}$ such that, among~$a$,~$b$, and~$c$,~$\Pi_a$ consists of all votes ranking~$a$ last,~$\Pi_b$
all votes ranking~$b$ last, and~$\Pi_c$ all votes ranking~$c$ last. Observe that~$\Pi_a$,~$\Pi_b$, and~$\Pi_c$ are all nonempty.
Moreover, if $(\mathcal{C}, \Pi_{\mathcal{V}})$ is $2$-axes single-peaked, then none of~$\Pi_{T}$ and~$\Pi_{F}$ contains three votes from~$\Pi_a$,~$\Pi_b$, and~$\Pi_c$,
respectively, where~$\Pi_{T}$ and~$\Pi_{F}$ are as discussed above. Due to symmetry, we have three cases to consider:
\begin{description}
\item[Case~1] $\Pi_a\subseteq \Pi_{T}$, $\Pi_b\subseteq \Pi_{F}$.
\item[Case~2] $\Pi_a\subseteq \Pi_{T}$, $\Pi_c\subseteq \Pi_{F}$.
\item[Case~3] $\Pi_b\subseteq \Pi_{T}$, $\Pi_c\subseteq \Pi_{F}$.
\end{description}
If $(\mathcal{C}, \Pi_{\mathcal{V}})$ is $2$-axes single-peaked, at least one of the above cases holds, and vice versa. Due to symmetry, we analyze only Case~1.
The analysis for the other two cases are similar. First, if $(\mathcal{C}, \Pi_a)$ or $(\mathcal{C}, \Pi_b)$ are not single-peaked,
we discard this case. Assume now that both $(\mathcal{C}, \Pi_a)$ and $(\mathcal{C}, \Pi_b)$ are single-peaked. Then, we create the clauses as follows.
We shall ensure that there is a partition $(\Pi_{T}, \Pi_F)$ of $\Pi_{\mathcal{V}}$  such that $(\mathcal{C}, \Pi_T)$, $(\mathcal{C}, \Pi_F)$ are single-peaked,
$\Pi_a\subseteq \Pi_T$ and $\Pi_b\subseteq \Pi_F$, if and only if the 2SAT instance has a truth assignment satisfying all the following clauses.
%We show only the 2SAT instance created for Case~1. The 2SAT instances for Cases~2 and~3 are analogously.

\begin{itemize}
\item If there is a vote $\pi\in \Pi_c$ which $\alpha$-conflicts with one vote in $\Pi_a$ (resp.\ $\Pi_b$), or
WD-conflicts with two votes in~$\Pi_a$ (resp.\ $\Pi_b$), then~$\pi$ must be included in~$\Pi_{F}$ (resp.\ $\Pi_{T}$).
In this case, we create a clause $(\overline{x(\pi)})$ (resp.\ $(x(\pi))$).%, to enforce that $f(v(\pi))=0$ in order to satisfy this clause.

\item If there are two votes $\pi, \pi'\in \Pi_c$ which WD-conflict with one vote in~$\Pi_a$ (resp.\ $\Pi_b$),
we create one clause $(\overline{x(\pi)}, \overline{x(\pi')})$ (resp.\ $({x(\pi)}, {x(\pi')})$), to ensure that at least one of $\{\pi, \pi'\}$ is in~$\Pi_{F}$ (resp.\ $\Pi_{T}$).
\end{itemize}

One may wonder that there might be three votes $\pi, \pi', \pi''\in \Pi_c$ that WD-conflict.
We don't need to consider this case since it has been implicitly dealt with in the second type of clauses. Assume that $(\pi, \pi', \pi'', d, d', d'')$ is a WD-structure,
where $d, d', d''\in \mathcal{C}$ and $\pi, \pi', \pi''\in \Pi_c$. Let~$\pi_1$ and~$\pi_2$ be two arbitrary votes from~$\Pi_a$ and~$\Pi_b$, respectively.
If the last ranked candidates among~$d$,~$d'$,~$d''$ in~$\pi_1$ and~$\pi_2$ are the same, say,~$d$, then,~$\pi'$ and~$\pi''$ WD-conflict with both~$\pi_1$ and~$\pi_2$.
Hence, two clauses~$(x(\pi'), x(\pi''))$ and~$(\overline{x(\pi')}, \overline{x(\pi'')})$ have been created according to the above discussion.
In order to satisfy these two clauses,~$x(\pi')$ and~$x(\pi'')$ must be assigned different values and the votes~$\pi'$ and~$\pi''$ are partitioned into different sets accordingly.
On the other hand, assume that the last ranked candidates among~$d$,~$d'$,~$d''$ in~$\pi_1$ and~$\pi_2$ are different.
Without loss of generality, assume that~$\pi_1$ ranks~$d$ in the last and~$\pi_2$ ranks~$d'$ in the last.
Then,~$\pi'$ and~$\pi''$ WD-conflict with~$\pi_1$, and~$\pi$ and~$\pi''$ WD-conflict with~$\pi_2$.
According to the above discussion, we have two clauses~$(\overline{x(\pi')}, \overline{x(\pi'')})$ and $({x(\pi)}, {x(\pi'')})$.
Again, to satisfy these two clauses,~$x(\pi)$,~$x(\pi')$, and~$x(\pi'')$ cannot be assigned the same value, leading to~$\pi$,~$\pi'$,~$\pi''$ not being partitioned into the same set.

Now we analyze the running time of the algorithm which is dominated by the following steps.

(1) Find a WD-structure and the three candidates~$a$,~$b$, and~$c$ as discussed above. This can be done in $\bigo{n^3\cdot m^3}$ time. The partition $(\Pi_a, \Pi_b, \Pi_c)$
as defined in the algorithm can be constructed in $\bigo{m\cdot n}$ time.

(2) Check if $(\mathcal{C}, \Pi_a)$ and $(\mathcal{C}, \Pi_b)$ are single-peaked, which can be done in $\bigo{m\cdot n}$ time~\cite{DBLP:conf/ecai/EscoffierLO08}.

(3) Find all pairs $(\pi, \pi')$ such that~$\pi$ and $\pi'$ $\alpha$-conflict and $\pi\in \Pi_c$ and $\pi'\in \Pi_a\cup\Pi_b$, and create a clause for each such pair. This can be done in
$\bigo{n^2\cdot m^4}$ by enumerating all possibilities.

(4) Find all $3$-tuple $(\pi, \pi', \pi'')$ such that $\pi$, $\pi'$, and $\pi''$ WD-conflict and $\pi\in \Pi_a$, $\pi'\in \Pi_b$, $\pi''\in \Pi_c$,
and create one clause for each $3$-tuple
as described above. This can be done in $\bigo{n^3\cdot m^3}$.

In summary, we need $\bigo{n^2\cdot m^3\cdot (n+m)}$-time to create $\bigo{n^2\cdot m^3\cdot (n+m)}$ clauses.
Then, the whole algorithm runs in $\bigo{n^2\cdot m^3\cdot (n+m)}$ time as the 2SAT can be solved in linear time in the number of clauses~\cite{DBLP:journals/siamcomp/EvenIS76}.
\end{proof}

%Now we turn our attention to the {\sc{$k$-Candidate Partition Recognition}} problem. We prove that the problem is {\fpt} with respect to $k$.
%Our {\fpt}-algorithm for the problem is in sprite of the dynamic-based polynomial-time algorithm for {\sc{$k$-Candidate Deletino Recognition}},
%proposed by Erd\'{e}lyi, Lackner, and Pfandler~\cite{ErdLP2016JAIR}. The {\sc{$k$-Candidate Deletino Recognition}} problem is to determine if
%there are $k$ candidates in an election whose deletion results in a single-peaked election.

\onlyfull{Erd\'{e}lyi, Lackner, and Pfandler posted in~\shortcite{Erdelyi2017} the open question whether determining an election is a $k$-CP election is {\nph}.
We resolve this question in the affirmative. Without stopping here, we further show that this problem is {\fpt} with respect to~$k$.
%Our {\nphns} reduction is from the X3C problem and the {\fpt}-algorithm is based dynamic programming and a nontrivial extension of the algorithm in for recognition single-peaked elections.
\begin{theorem}
%{\sc{$k$-Candidate Partition Recognition}} is {\fpt} with respect to~$k$.
Determining whether an election is a $k$-CP election is {\nph}, but {\fpt} with respect to~$k$.
\end{theorem}
}

%%%%%%%%%%%%%%%%%%%%%%%%%%%%%%%
%%%%%%%%%%%%%%%%%%%%%%%%%%%%%%%
%%%%%%%%%%%%%%%%%%%%%%%%%%%%%%%
\section{Conclusion}
Aiming at pinpointing the complexity border of CCAV and CCDV between single-peaked elections and general elections,
we have studied these problems in $k$-axes elections and $k$-CP elections and obtained many tractability and intractability results.
We particularly studied the voting correspondences $r$-approval, Condorcet, Maximin, and Copeland$^{\alpha}, 0\leq \alpha\leq 1$. Our study closed many gaps left in the literature. 
We refer to Table~\ref{tab_our_results} for a summary of our results. 
%Our study leads to several dichotomy results for CCAV and CCDV for the above voting correspondences in $k$-axes, $k$-CP and $k$-peaked elections with respect to the values of~$k$.
Though that our focus in this paper is the unique-winner model of CCAV and CCDV, it should be pointed out that all our results, including polynomial-time solvability results,
{\fpt} results, and {\nphns} results hold for the nonunique-winner model of CCAV and CCDV as well. Recall that in the nonunique-winner model,  the goal is to make the distinguished
candidate~${\discandi}$ a winner, but not necessarily the unique winner. %A general conclusion is that CCAV and CCDV are {\nph} even when $k=2$, with the exception
%of the problems for Condorcet which are {\fpt} with respect to~$k$.

In addition, we proved that determining whether an election is a $k$-axes election is polynomial-time solvable for $k=2$. Given that the problem is polynomial-time solvable
for $k=1$~\cite{DBLP:journals/scw/BallesterH11} and {\nph} for every $k\geq 3$~\cite{Erdelyi2017}, our result closes the final complexity gap of the problem with respect to~$k$.
%Finally, we show that determining whether an election is a $k$-CP election is {\nph} but {\fpt} with respect to~$k$, resolving an open question posted in~\cite{Erdelyi2017}.

%It should be pointed out that our {\fpt} results for $k$-axes elections extend to some other domains of nearly single-peaked

There remain several open questions (see Table~\ref{tab_our_results}) for future research. In particular, the complexity of CCAV and CCDV for Copeland$^{\alpha}$, $0\leq \alpha<1$,
restricted to single-peaked elections remained open.
In addition, investigating the complexity of other voting problems, e.g., {\sc{Destructive Control By Adding/Deleting Votes/Candidates, Bribery}},
in nearly single-peaked elections would be another promising topic for future research.

\section*{Acknowledgement}
The author would like to thank the anonymous reviewers of ECAI 2020 for their constructive comments.
A $3$-page extended abstract of this paper appeared in the proceedings of AAMAS 2018~\protect\cite{DBLP:conf/atal/Yang18}.
This version has a much expanded introduction section, provides many proofs, and resolves some open questions left in the previous version.

%\onlyfull
{
%\newpage
\section*{Appendix}
In this section, we provide the proof of Lemma~\ref{lem-3-regular-bipartite-decomposable}.

%\begin{lemma}
%Let~$G$ be a $3$-regular bipartite graph with vertex set~$\vset{G}$ and edge set~$\eset{G}$.
%Then, there are two linear orders~$\lhd_1$ and~$\lhd_2$ over~$\vset{G}$ and a partition $(A_1, A_2)$
%of~$\eset{G}$ such that for every $i\in \{1,2\}$ and for every edge $\edge{u}{v}\in A_i$, it holds that~$u$ and~$v$ are consecutive in~$\lhd_i$.
%\end{lemma}

\begin{proof}
Let~$G$ be a $3$-regular bipartite graph with vertex set~$\vset{G}$ and edge set~$\eset{G}$.
Due to a classic theorem of K\"{o}nig~\shortcite{Koenig1916},~$\eset{G}$ can be partitioned into three perfect matchings~$M_1$,~$M_2$,~and~$M_3$ of~$G$. 
Let~$H_1$ be the graph with vertex set~$\vset{G}$ and edge set $M_1\cup M_2$, and let~$H_2$ be the graph with vertex set~$\vset{G}$ and edge set~$M_3$.
As~$M_1$ and~$M_2$ are perfect matchings, each vertex in~$H_1$ is of degree~$2$.
In other words,~$H_1$ consists of vertex-disjoint cycles. We shall remove some edges from~$H_1$ to~$H_2$
so that both~$H_1$ and~$H_2$ consist of only vertex-disjoint paths. Obviously,~$H_2$ is already such a graph. We use the following procedure to suit our purpose.

\begin{procedure}
\ForEach{cycle~$L$ in~$H_1$\;}
{
Let~$\edge{u}{w}$ be an edge in~$L$ such that~$u$ and~$w$ are in different paths and are both of degree-$1$ in~$H_2$\;\label{line-hl}
 Remove the edge~$\edge{u}{w}$ from~$H_1$ into~$H_2$\; \label{line-h2}
}
\end{procedure}

Clearly, after adding an edge as described in Line~\ref{line-h2} in the above procedure into~$H_2$, the graph~$H_2$ still consists of only vertex-disjoint paths.
After removing one edge in a cycle, the cycle becomes a path. So, after the above procedure,~$H_1$ consists of only vertex-disjoint paths too.
The proof is not completed because we have not shown that when a cycle in~$H_1$ is considered in the above procedure,
there always exists an edge as required in Line~\ref{line-hl}. We shall prove this statement now by induction.
Let $L_1, L_2, \dots, L_t$ be the vertex-disjoint cycles in~$H_1$, and without loss of generality, assume that~$L_i$ is considered before~$L_{i+1}$ in the procedure for each $i\in [t-1]$.
Moreover, for ease of exposition, for each $i\in [t]$, let~$H_2^i$ denote the graph~$H_2$ after exactly~$i$ cycles in~$H_1$ are considered.
So, the edge set of~$H_2^1$ is exactly~$M_3$.
As~$M_1$,~$M_2$, and~$M_3$ are disjoint perfect matchings of~$G$, for any edge in~$L_1$ the two endpoints of this edge are in different paths and are of degree~$1$ in~$H_2^1$.
Assume that, for some $i\in [t]$, after cycles $L_1, L_2, \dots, L_{i-1}$ are dealt with,~$H_2^{i-1}$ consists of vertex-disjoint paths. Consider now the cycle~$L_{i}$.
We first have the following claim.
\medskip

{\textbf{Claim.}} Every vertex in the circle~$L_i$ is a degree-$1$ vertex in $H_2^{i-1}$.
\smallskip

We prove the claim as follows. Assume for the sake of contradiction there is a vertex~$u$ in~$L_i$ who has at least two neighbors in the graph~$H_2^{i-1}$.
Recall that all cycles in~$H_1$ are vertex-disjoint, and hence~$u$ does not belong to any of~$L_1, L_2, \dots, L_{i-1}$.
Due to the above procedure, we know that all edges incident to~$u$ in~$H_2^{i-1}$ must be from~$M_3$.
However, if~$u$ has at least two neighbors in~$H_2^{i-1}$,~$M_3$ cannot be a matching, a contradiction.
\smallskip

Based on the above claim, we find a desired edge in~$L_i$ as follows.
Let~$\edge{u}{w}$ be an arbitrary edge in~$L_i$. If $\edge{u}{w}$ fulfills the condition in Line~\ref{line-hl}, we are done.
Otherwise, due to the above claim, it must be that~$u$ and~$w$ are the two ends of a path in~$H_2^{i-1}$.
Let~$w'$ be the other neighbor of~$w$ in the cycle~$L_i$.
Due to the above claim,~$w'$ has degree~$1$ in $H_2^{i-1}$. This directly means that~$w'$ cannot be in the same path as~$u$ and~$w$.
Then,~$\edge{w}{w'}$ is the desired edge which can be removed from~$L_i$ to $H_2^{i-1}$ as described in
Line~\ref{line-h2}.

Now, let~$H_1^t$ and~$H_2^t$ be the graphs returned by the above procedure.
Then, we can define the desired linear orders~$\lhd_1$ and~$\lhd_2$ based on the fact that both~$H_1^t$ and~$H_2^t$ consists of only vertex-disjoint paths.
Precisely, let $(P_1, P_2, \dots, P_t)$ be an arbitrary but fixed order of all paths in~$H_1^t$. For each~$P_i$ where $1\in [t]$, let $(P_i(1), P_i(2), \dots, P_i(\ell))$ denote
the path~$P_i$, i.e.,~$\ell$ is the number of vertices in the path and there is an edge between~$P_i(j)$ and~$P_i(j+1)$ for all $j\in [\ell-1]$ in~$H_1^t$.
Then, we define~$\lhd_1$ so that all vertices in~$P_i$ are ordered before all vertices in~$P_{i+1}$ for all $i\in [t-1]$.
Moreover, the relative order of the vertices in each~$P_i$ in~$\lhd_1$ is exactly $(P_i(1), P_i(2), \dots P_i(\ell))$.
The linear order~$\lhd_2$ can be constructed analogously.
\end{proof}
}%

\endgroup

%%\def\url#1{}  %% chaged in ACM-REference-format.bst
%%\bibliographystyle{ACM-Reference-Format}
%%\bibliographystyle{ecai}
%\bibliographystyle{yangacmsmall}
%%%\bibliographystyle{plain}
%\bibliography{../../Bibs/sociachoiceref,../../Bibs/graphref}

\end{document}